\documentclass[11pt]{article}
\usepackage[margin=1in]{geometry}

\usepackage{verbatim}
\usepackage{tikz}
\usetikzlibrary{shapes.geometric}
\usetikzlibrary{quotes}
\usetikzlibrary{arrows.meta}
\usepackage{subfig}

\usepackage{graphicx} % Required for inserting images
\usepackage[english]{babel}
\usepackage{amsmath,amsthm,mathtools}

\title{Approximation Schemes for Orienteering and Deadline TSP in Doubling Metrics}

\author{Kinter Ren\\
Department of Computer Science\\
University of Alberta\\
zr4@ualberta.ca
\and Mohammad R. Salavatipour\footnote{Supported by NSERC}\\
Department of Computer Science\\
University of Alberta\\
mrs@ualberta.ca}
\date{}

\theoremstyle{plain}
\usepackage{hyperref}
\usepackage{bbm}
\usepackage{amsmath,amssymb,enumerate,ifthen,tikz,floatrow,amsmath, algorithm,algorithmic}
\usepackage{mathtools}
\usepackage{graphicx}

\graphicspath{ {./images/} }
\newtheorem{theorem}{Theorem} 
\newtheorem{lemma}{Lemma} 
 
\newtheorem{proposition}{Proposition}

\newtheorem{definition}{Definition} 

\newtheorem{remarka}{Remark}

\usepackage{float}

\usepackage{hyperref}

\newcommand{\QQ}{\mathbb{Q}}

\newcommand{\eps}{\varepsilon}
\newcommand{\Poly}{\text{Poly}}
\newcommand{\calE}{{\cal E}}

\newcommand{\calI}{{\cal I}}
\newcommand{\calQ}{{\cal Q}}
\newcommand{\E}{{ \mathbb{E}}}
\newcommand{\opt}{{\text{opt}}}

\begin{document}

\maketitle

\begin{abstract}
In this paper we look at $k$-stroll, point-to-point orienteering, as well as the deadline TSP  problem on graphs with bounded doubling dimension and bounded treewidth and present approximation schemes for them.
These are extensions of the classic Traveling Salesman Problem (TSP) problem. Suppose 
we are given a weighted graph $G=(V,E)$ with  a start node $s\in V$, distances on the edges $d:E\rightarrow\QQ^+$ and integer $k$. In the $k$-stroll problem the goal is to find a path starting at $s$ of minimum length that visits at least $k$ vertices. The dual problem to $k$-stroll is the rooted orienteering in which instead of $k$ we are given a budget $B$ and the goal is to find a walk of length at most $B$ starting at $s$ that visits as many vertices as possible. In the point-to-point orienteering ({P2P orienteering}) we are given start and end nodes $s,t$ and the walk is supposed to start at $s$ and end at $t$. In the deadline TSP (which generalizes {P2P orienteering}) we are given a deadline $D(v)$ for each $v\in V$ and the goal is to find a walk starting at $s$ that visits as many vertices as possible before their deadline (where the visit time of a node is the distance travelled from $s$ to that node).
The  best approximation for rooted orienteering (or {P2P orienteering}) is $(2+\eps)$-approximation \cite{chekuri2012improved} and $O(\log n)$-approximation for deadline TSP \cite{Bansal04}.
For Euclidean metrics of fixed dimension, Chen and Har-peled present \cite{chen2008euclidean} a PTAS for rooted orienteering. There is no known approximation scheme for deadline TSP for any metric (not even trees). Our main result is the first approximation scheme for deadline TSP on metrics with bounded doubling dimension (which includes Euclidean metrics). To do so we first
we present a quasi-polynomial time approximation scheme for $k$-stroll and {P2P orienteering} on such metrics. More specifically, if $G$ is a metric with  doubling dimension $\kappa$ and aspect ratio $\Delta$, we present a $(1+\eps)$-approximation that runs in time $n^{O\left(\left(\log\Delta/\eps\right)^{2\kappa+1}\right)}$.
We then extend these to obtain an approximation scheme for deadline TSP when the distances and deadlines are integer which runs in time 
$n^{O\left(\left(\log \Delta/\eps\right)^{2\kappa+2}\right)}$. The same approach also implies a bicriteria $(1+\eps,1+\eps)$-approximation for deadline TSP for when distances (and deadlines) are in $\QQ^+$.
For graphs with bounded treewidth $\omega$ we show how to solve {$k$-stroll} and {P2P orienteering} exactly in polynomial time and a $(1+\eps)$-approximation for deadline TSP in time $n^{O((\omega\log\Delta/\eps)^2)}$.
\end{abstract}

%%%%%%%%%%%%%%%%%%%%%%%%%%%%%%%%%%%%%%%%%%%%%%%%5
\section{Introduction}
\label{sec:intro}
We study a fundamental variant of the \emph{Traveling Salesman Problem} ({\bf TSP}) in which the ``salesman'' wants to visit as many customers as possible before their required service \emph{deadlines}. Suppose we are given a weighted graph $G=(V,E)$ with a start node $s\in V$, deadlines $D:V\rightarrow \QQ^+$, and a distance/length function $d: E\rightarrow \QQ^+$.
In the {\em deadline TSP} problem, our objective is to find a path (or a walk) starting at $s$ that visits as many nodes as possible by their deadlines.
We say that the path visits node $v$ by its deadline if the length of the path, starting at $s$ until it visits $v$ (for the first time), is at most $D(v)$. A node might be visited multiple times, but it is counted only once. Alternatively, we can assume $G$ is a complete weighted graph where the edge weights satisfy triangle inequality, i.e. the $(V,d)$ is a metric space.

A special case of the deadline TSP is the \emph{rooted Orienteering} problem, where
all the nodes have a universal deadline $B$ (think of it as a budget for the length
of the path travelled) and the goal is to find a path, starting at $s$, of length at most
$B$ that visits as many vertices as possible. Researchers have also studied another
version of the Orienteering, known as
\emph{point-to-point Orienteering} problem, denoted by {\em P2P orienteering}, where we are given a \emph{start node} $s \in V$, an \emph{end node} $t \in V$, and a \emph{budget} $B$. Here, our objective is to find an $s$-$t$ path of length at most $B$ that visits as many nodes (other than the end-points) as possible. This too can be viewed as a special case of deadline TSP: the deadlines for each node $v$ are set as $B - d(v, t)$, and the deadline TSP path is completed by connecting it directly to the end node $t$. Notice that {P2P orienteering} can be seen as the ``dual'' of the classical {\em $k$-stroll} problem \cite{blum2007approximation} (also known as the \emph{$k$-path} problem \cite{chaudhuri2003paths}; and the (path) $k$-TSP problem \cite{chen2008euclidean}), wherein the objective is to find a minimum-length $s$-$t$ path visiting at least $k$ nodes. We denote this problem as {$k$-stroll}.
It is important to note that {P2P orienteering} and {$k$-stroll} are equivalent in terms of exact algorithms. However, an approximation for one problem does not imply an approximation for the other. 

Blum et al. \cite{blum2007approximation} developed the first O(1)-approximation algorithm for the rooted orienteering problem (with no fixed end node) in a general metric space.
Their main approach involves a clever reduction to $k$-TSP,  transforming an approximation algorithm for the $k$-TSP problem into an approximation algorithm for the rooted orienteering problem via an intermediate problem: They introduced the concept of \emph{excess} (also referred to as \emph{regret}) for a path and formulated the related \emph{minimum-excess} problem, demonstrating that an approximation for the latter extends to an approximation for the orienteering problem. Given a path $P$ that includes two vertices $u,v$, let $d_P(u,v)$ denote the length of the sub-path from $u$ to $v$.
The excess of an $s,t$-path $P$, denoted by ${\cal E}_P$, is the difference between the length of the path and the distance between the endpoints of the path, ${\cal E}_P=d_P(s,t)-d(s,t)$ and the minimum-excess problem seeks to find an $s$-$t$ path of minimum-excess that visits at least $k$ nodes. 
Moreover, they \cite{blum2007approximation} demonstrated that the min-excess problem can be approximated using algorithms for {$k$-stroll}, implying that an approximation algorithm for {$k$-stroll} leads to an approximation algorithm (with a constant factor loss in the ratio) for the orienteering problem, via the intermediate min-excess problem.
Applying a similar approach and concentrating on enhancing the approximation algorithm for the {$k$-stroll}, the factor has been refined to the currently best $(2 + \epsilon)$-approximation, which also extends to {P2P orienteering} \cite{chekuri2012improved} in a general metric space.

Chen and Har-Peled \cite{chen2008euclidean} developed the first Polynomial-Time Approximation Scheme (PTAS) for the rooted orienteering problem in a fixed-dimensional Euclidean space. The prior reduction by \cite{blum2007approximation} increases the approximation ratio by a constant factor and can not provide a PTAS.
 To achieve a PTAS for the rooted orienteering problem, Chen and Har-Peled \cite{chen2008euclidean} introduced the concept of $\mu$-excess for a path (extending the notion of the excess of a path) and defined $(\epsilon, \mu)$-approximation for $k$-TSP. Roughly speaking, the $\mu$-excess of a path is the difference between its length and the length of the best approximation of the path using straight lines between only $u$ nodes of the path (as opposed to considering the length of the line connecting the two endpoints of the path, as in the min-excess path). This notion is formally defined soon.
 An $(\epsilon, \mu)$-approximation for rooted $k$-TSP is any path that starts from the root and visits $k$ points and whose length is at most $||P^*||$ + $\epsilon \cdot \calE_{P^*, \mu}$, where $||P^*||$ denotes the length of the optimal rooted $k$-TSP path and $\calE_{P^*, \mu}$ is the $\mu$-excess for the optimal path. 
 Note that the upper bound provided by $(\epsilon, \mu)$-approximation for the $k$-TSP may be significantly tighter than that by $(1+\epsilon)$-approximation (namely, $(1+\epsilon) \cdot ||P^*||$), especially when $\mu$ is sufficiently large, as the $\mu$-excess of a path can be much smaller than the length of the path. In fact, Chen and Har-Peled \cite{chen2008euclidean} showed that  the same algorithm that Mitchell developed as a PTAS for classic TSP on Euclidean plane \cite{mitchell1999guillotine} provides a $(\epsilon, \frac{1}{\epsilon})$-approximation for the rooted $k$-TSP problem on the that metric and this leads to a PTAS for the rooted orienteering problem.

Deadline TSP seems substantially more difficult to approximate.
For deadline TSP Bansal et al. \cite{Bansal04} presented an $O(\log n)$-approximation and extended that to an $O(\log^2 n)$-approximation for the Time-Window TSP problem (where each node $v$, other than a deadline $D(v)$, has also "release" time $R(v)$, and it is "counted" if the time it is visited along the output path falls in the interval $[R(v),D(v)]$). They also provided a bicriteria approximation: one that produces a $\log(1/\eps)$-approximate solution that violates the deadlines by $(1+\eps)$. The $O(\log n)$-approximation remains the best known approximation for deadline TSP on general metrics (in polynomial time).
More recently, Friggstad and Swamy \cite{FriggstadS22} presented an
$O(1)$-approximation for deadline TSP which runs in time $O(n^{\log n\Delta})$ where $\Delta$ is the diameter of the graph; for this they assume that distances (and so deadlines) are all integers.

%In other words, if one assumes that the distances (and deadlines) are all values bounded in $\Poly(n)$, this
%implies a quasi-polynomial time $O(1)$-approximation. 
Approximations for both orienteering and deadline TSP for other metrics (such as directed graphs) have been obtained as well (discussed in Section \ref{sec:related}). For instance, for deadline TSP on trees Farbstein and Levin \cite{FV19}, present a bicriteria $(1+\eps,1+\eps)$-approximation (i.e. $(1+\eps)$-approximation and violating the deadlines by $(1+\eps)$).
To the best of our knowledge, there is no true approximation scheme for deadline TSP on any metrics (even on Trees). 
Our goal in this work has been to obtain an approximation scheme for
deadline TSP on bounded doubling dimensions (which implies the first approximation scheme for Euclidean metrics as well). 
 
%%%%%%%%%%%%%%%%%%%%%%%%%%%%%%%%%%%%%%%%%%%%%%%%%
\subsection{Our Results and Techniques.}
We present a number of results for P2P Orienteering and deadline TSP.
By scaling we assume all distances are at least 1 and
at most $\Delta$.
As a starting point we show that orienteering can be solved exactly in polynomial time on graphs with bounded treewidth and how this can be turned into a QPTAS for deadline TSP on such metrics when the distances are integer and $\Delta$ is quasi-polynomial in $n$.

\begin{theorem}\label{thm:boundedTW}
  Given a graph with treewidth $\omega$,
  there is an exact algorithm with running time is $O(\Poly(n)\cdot \omega^{\omega^2})$ for solving the {$k$-stroll} or P2P orienteering problem. Furthermore, there is an approximation scheme for deadline TSP on such graphs with integer distances with running time $n^{O((\omega\log\Delta/\eps)^2)}$.
 \end{theorem}

Proof of this theorem is simpler and has some of the main ideas used in the subsequent theorems; so is a good starting point.
Our second (and more significant) set of results are for metrics with bounded doubling dimension. These metrics include Euclidean metrics of fixed dimension. 
\begin{theorem}\label{thm:kstrollDBL}
Given a graph $G=(V,E)$ with doubling dimension $\kappa$, start and end nodes $s,t\in V$ and integer $k$. We can get a $(1+\eps)$-approximation for {$k$-stroll} in 
time $n^{O\left(\left(\log\Delta/\eps\right)^{2\kappa+1}\right)}$..
\end{theorem}

We actually show a stronger bound on the length of the {$k$-stroll} solution produced by this theorem.
We then use that to prove the following:

\begin{theorem}\label{thm:P2PDBL}
    Given a graph $G=(V,E)$ with doubling dimension $\kappa$, start and end nodes $s,t\in V$ and length budget $B$. We can get a $(1+\eps)$-approximation for P2P orienteering in time $n^{O\left(\left(\log\Delta/\eps\right)^{2\kappa+1}\right)}$.
\end{theorem}
The reason that our running time 
in Theorems \ref{thm:kstrollDBL} and \ref{thm:P2PDBL} includes $\log \Delta$ in the exponent is that our algorithm uses the (randomized) hierarchical decomposition of such metrics \cite{Talwar04}, which produces a decomposition of height $O(\log \Delta$), and so the algorithm has a running time with $n^{\Poly(\log n)\log\Delta}$. 
We should note that for most vehicle routing problems on Euclidean or doubling metrics, one can assume that $\Delta=\Poly(n)$ by a standard scaling of distances at a loss of $(1+\eps)$ in approximation factor (for e.g. this has been the first step in \cite{Arora98,Talwar04,BartalGK16} for designing approximation schemes for TSP on Euclidean and doubling metrics). 
This however does not work for budgeted versions (such as orienteering or deadline TSP) as a feasible solution for the scaled version might violate the budget(s) by  a $(1+\epsilon)$ factor unless one is settling for a bi-criteria approximation. 

Building upon these theorems we present the strongest result of this paper. Assuming that distances are all integers, building upon the results developed above we show how one can get a $(1+\eps)$-approximation for deadline TSP (without any violation of deadlines) on doubling metrics. This is the first approximation scheme for any non-trivial metrics for deadline TSP.

\begin{theorem}\label{thm:deadlineDBL}
    Suppose $G=(V,E)$ is a graph with doubling dimension $\kappa$, 
    start nodes $s\in V$ as an instance of deadline TSP where distances (and deadlines) are integer values. We can get a $(1+\eps)$-approximation for deadline TSP in time $n^{O((\log\Delta/\eps)^{2k+2})}$.
\end{theorem}

We note that the algorithm of this theorem can be adapted to 
obtain a bicriteria $(1+\eps,1+\eps)$-approximation for deadline TSP (i.e. cost at most $1+\eps$ the optimum while violating the deadlines by at most $1+\eps$) that works even when distances are in $\QQ^+$ (instead of integer) with the same running time.
%Note that for the case of Euclidean metrics (which have bounded doubling dimension) even if the input points are assumed to have integer coordinates, the distances can be real numbers. Therefore, calculation of precise distances on a path (to compare with deadlines) is prone to errors in numerical calculations, unless one places some assumptions on (approximate) values of distances and deadlines.

We start with some preliminaries in Section \ref{sec:prel}.
We prove Theorem \ref{thm:boundedTW} in Section \ref{sec:poly_p2p}, Theorems \ref{thm:kstrollDBL} and \ref{thm:P2PDBL} in Section \ref{sec:qptas_p2p}, and finally
Theorem \ref{thm:deadlineDBL} in Section \ref{sec:qptas_deadlineTSP}.

%%%%%%%%%%%%%%%%%%%%%%%%%%%%%%%%%%%%%%%%%%%%%%
\subsection{Related Works}\label{sec:related}

The seminal works of Arora \cite{Arora98} and Mitchell \cite{mitchell1999guillotine} presented the first PTAS
for Euclidean TSP and some of its variants such as $k$-TSP. However, extending this to orienteering proved to be challenging. 
For orienteering on Euclidean plane Arkin et al. \cite{Arkin98} presented a $(2+\eps)$-approximation.
The first PTAS for rooted orienteering on Euclidean metrics was given by Chen and Har-peled \cite{chen2008euclidean}.
More recently Gottlieb et al. \cite{GottliebKR22} presented a more efficient PTAS for rooted orienteering on Euclidean metrics.

For orienteering on general metrics there are several $O(1)$-approximations \cite{blum2007approximation,chekuri2012improved,FriggstadS17}, with the current best being $(2+\eps)$-approximation by \cite{chekuri2012improved}.
For orienteering on directed graphs Chekuri and P{\'{a}l} \cite{ChekuriP05} present a quasi-polynomial time for various problems including deadline TSP, such as an $O(\log \opt)$-approximation for time window in quasi-polynomial time. Nagarajan and Ravi \cite{NagarajanR11} also present poly-logarithmic approximations for directed variants of $k$-TSP and Orienteering, some of these results were improved in \cite{BateniC}.
Chekuri et al. \cite{chekuri2012improved} present an $O(\log^2 \opt)$-approximation for the time-window version and an $O(\log \opt)$-approximation for {P2P orienteering} on directed graphs. There is also a line of research on stochastic orienteering (e.g. see \cite{Stochastic15, BansalN13, Jiang20}), with a constant factor approximation known for non-adaptive and $O(\log\log B)$-approximation for adaptive policies \cite{Stochastic15}.

As said earlier, the best known result for deadline TSP on general metrics is by \cite{Bansal04} which is
an $O(\log n)$-approximation. They also provide a bicriteria $(O(\log\frac{1}{\eps}),1+\eps)$. Assuming
distances (and deadlines) are integers, this implies an $O(\log D_{\max})$-approximation where $D_{\max}$ is the maximum deadline. Different extensions of deadline TSP have been studied also.
For instance, Bansal et al. \cite{Bansal04} present an $O(\log^2 n)$-approximation for the time-window version and an $O(\log D_{\max})$-approximation when the time window values are integer.
Chekuri et al. \cite{chekuri2012improved} prove that any $\alpha$-approximation for {P2P orienteering} implies
an $O(\alpha\max{\{\log \opt,\log \frac{L_{\max}}{L_{\min}}\}})$-approximation for
the time-window version where $\opt$ is the optimal value and $L_{\max}$ and $L_{\min}$ are the sizes of the largest and smallest windows. 
For the variant of time-window TSP where there are multiple time-windows given for each node, Garg et al. \cite{Gargkk21} proved an 
$\Omega(\frac{\log\log n}{\log\log\log n})$-hardness of approximation even on trees.
There are other variants of orienteering and deadline TSP that have been studied (see \cite{chekuri2012improved,ChekuriP05,FriggstadS22,GaoJMZ16}). 

Some variants of vehicle routing problems have been studied on doubling metric as well. Talwar \cite{Talwar04} presented a QPTAS for TSP on such metrics. Bartal et al. \cite{BartalGK16} building upon the work of Talwar presented a PTAS for TSP on doubling metrics. Several variants of TSP, such as Capacitated Vehicle Routing Problem (CVRP) or TSP with neighborhoods  have approximation schemes on doubling metrics \cite{ChanJ16,JS23}.

%%%%%%%%%%%%%%%%%%%%%%%%%%%%%%%%%%%%%%%%%%%%%%%%%%%%%%%%55
\section{Preliminaries}\label{sec:prel}
We consider a metric space $(V, d)$ as a (complete) weighted graph $G=(V,E)$ and by scaling we assume the minimum pair-wise distance is 1. For $v \in  V$ and $r  \geq  0$, we let $B(v, r) = \{ u \in  V~ |~  d(u, v) \leq  r \}$ denote the ball of radius $r$ around $v$.
 The \emph{doubling dimension} \cite{gupta2003bounded} of a metric space $(V, d)$ is the smallest value $\kappa$ such that for all $v \in  V$, for all $\rho  > 0$, every ball $B(v, 2\rho)$ can be covered by the union
of at most $2^{\kappa}$ balls of the form $B(z, \rho )$, where $z \in  V$. A metric is called \emph{doubling} when its doubling dimension is a constant.
For a subset $U \subseteq V$, the \emph{diameter} of $U$, denoted by $\Delta_U$, is defined as $\max_{u,v \in U} d(u,v)$.
%, and the aspect ratio of $U$ is given by $\frac{\max_{u,v \in U} d(u,v)}{\min_{u,v \in U} d(u,v)}$. 

\begin{definition}
A \textbf{tree decomposition} of a graph $G$ is a pair $(T, \{B_t\}_{t \in V(T)})$, where $T=(V',E')$ is a tree whose every node $t \in V'$ is assigned a vertex subset $B_t \subseteq V(G)$, called a bag, such that the following three conditions hold: 
\begin{enumerate}
    \item $\cup_{t \in V(T)} B_t = V(G)$. In other words, every vertex of $G$ is in at least one bag. 
    \item For every $uv \in E(G)$, there exists a node $t$ of $T$ such that bag $B_t$ contains both $u$ and $v$. 
    \item For every $u \in V(G)$, the set $T_u = \{t \in V(T) : u \in B_t \}$, i.e., the set of nodes whose corresponding bags contain $u$, induces a connected subtree of $T$. 
\end{enumerate}
\end{definition} 
The width of a tree decomposition is one less than the size of the largest bag of it.
A graph $G = (V,E)$ has treewidth $w$ if $w$ is the smallest width of a tree decomposition of $G$.

We now formally define the problems. We consider a metric space $(V, d)$ which is induced by an edge-weighted graph $G = (V, E)$. Let $P = \langle v_1, v_2, \ldots, v_k \rangle$ be a path in $G$, starting at $s=v_1$ and ending at $t=v_k$. The \emph{length of the path} is denoted by $||P|| = \sum_{i=1}^{k-1} d(v_i, v_{i+1})$. We use $|P|$ to denote the number of vertices on the path. If $P$ is a walk then $|P|$ is the number of distinct nodes in $P$. Let a \emph{$\mu$-jump of $P$} be a path $\langle v_1=v_{i_1}, v_{i_2}, \ldots, v_{i_{\mu}}=v_k \rangle$ of size $\mu \leq k$ obtained from $P$ by bypassing some of its intermediate points. 
The optimum $\mu$-jump of $P$, denoted by $J^*_\mu(P)$, is the $\mu$-jump of $P$ with the maximum length. We define the \emph{$\mu$-excess of $P$} (see \cite{chen2008euclidean}), denoted by $\calE_{P, \mu}$, to be the difference between the length of $P$ and the length of $J^*_\mu(P)$, i.e. $\calE_{P,\mu}=||P||-||J^*_\mu(P)||$. Observe that $\calE_{P, \mu}$ may be significantly smaller than $||P||$, especially when $\mu$ is sufficiently large.

In {$k$-stroll}, we are given a \emph{start node} $s$, an \emph{end node} $t$, and a \emph{target integer} $k$. The goal is to find a minimum length path from $s$ to $t$ that visits at least $k$ nodes.
An $(\eps,\mu)$-approximation for {$k$-stroll} is an algorithm that finds a path $P$ with $|P|=k$ where $||P||\leq ||P^*||+\eps\cdot\calE_{P^*,\mu}$ where $P^*$ is the optimum solution.

We study the following extension of {$k$-stroll} to \emph{multi-path} in Sections \ref{sec:poly_p2p} and \ref{sec:qptas_p2p} where we design our algorithms for {P2P orienteering}. 

\begin{definition}[{\bf multi-path $k$-stroll}]
Given a graph $G= (V, E)$, $m$ pairs of nodes $(s_i,  t_i)$, $1\leq i \leq m$, and a target integer $k$, the goal of {\bf multi-path $k$-stroll} is to find $m$ paths $P_i$ from $s_i$ to $t_i$ in $G$ such that:
\begin{itemize}
  \item the total number of distinct nodes visited by all paths $P_1, P_2, \cdots, P_m$ is at least $k$. (a node might be visited by multiple paths, but it is counted only once.)
  \item the total length of all $m$ paths $\sum_{i=1}^{m}||P_i||$ is minimized.
\end{itemize}
\end{definition}
Notice that {multi-path $k$-stroll} becomes {$k$-stroll} when $m=1$. 
In {P2P orienteering}, we are given a \emph{start node} $s$, an \emph{end node} $t$, and a \emph{budget} $B$. The objective is to find an $s$-$t$ path $P$ with $||P|| \leq B$ that maximizes the count of distinct nodes visited. 

For any path $P$ in $G$ and nodes $u, v$ in $P$, let $P_{uv}$ denote the \emph{subpath from $u$ to $v$} in $P$. In  deadline TSP, we are given a \emph{start node} $s \in V$ and \emph{deadlines} $D(v)$ for each node in $V$. A feasible solution is a path starting at $s$. We say that such a path visits node $v \in P$ by its deadline if $||P_{sv}||\leq D(v)$. The objective is to find a path $P$ starting at $s$ to maximize the count of distinct nodes visited within their deadlines. 

For graphs of bounded doubling dimension, like earlier works on problems on such metrics, we rely on a hierarchical decomposition of $V$. Suppose $(V,d)$ is a doubling metric with diameter $\Delta$. We employ the hierarchical decomposition of metric spaces by means of probabilistic partitioning. The decomposition is essentially the one introduced by Bartal~\cite{bar96}, and subsequently used by others. In particular, Talwar \cite{Talwar04} used this to design the first approximation scheme (QPTAS) for TSP and other problems in doubling metrics.
A \emph{cluster} $C$ in the metric $(V, d)$ is a subset of nodes in $V$. A \emph{decomposition} of the metric $(V, d)$ is a partitioning of $V$ into clusters.
A hierarchical decomposition of $V$ is a sequence of partitions of $V$, where each partition is a refinement of the previous one. Normally this is represented by a tree $T$ (called the {\it split-tree}), where each node of $T$ 
corresponds to a cluster. We use $C$ to both refer to a node in $T$ as well as the cluster (set of nodes in $V$) it corresponds to. The root node of $T$ corresponds to the single set $\{V\}$ and the leaf nodes correspond to singleton sets $\{\{v\}\}_{v\in V}$. The children of each node $C\in T$ correspond 
to a partition of $C$ where each part has diameter about half of that of $C$. 
The union of all subsets corresponding to the vertices at each level of $T$ constitutes a partition of $V$.

\begin{theorem}[\cite{Talwar04}]\label{thm:hier-dec}
There is a hierarchical decomposition of $V$, i.e. a sequence of partitions $\Pi_0$, $\Pi_1$, $\dots, \Pi_h$, 
where $\Pi_{i-1}$ is a refinement of $\Pi_i$, $\Pi_h=\{V\}$, and $\Pi_0=\{\{v\}\}_{v\in V}$, and satisfies the following:
\begin{enumerate}
\item $\Pi_0$ corresponds to the leaves and $\Pi_h$ corresponds to the root of the split-tree $T$, and height of
 $T$ is $h=\delta+2$, where $\delta = \log \Delta$ and $\Delta$ is the diameter of the metric.
\item For each $C\in T$ at level $i$, cluster $C\in\Pi_i$, has diameter at most $2^{i+1}$.
\item The branching factor $b$ of $T$ is at most $2^{O(\kappa)}$.
\item For any $u,v \in V$, the probability that they are in different sets corresponding to nodes in level $i$ of $T$ is at most
 $O(\kappa)\cdot \frac{d(u,v)}{2^i}$. 
\end{enumerate}
\end{theorem}

%%%%%%%%%%%%%%%%%%%%%%%%%%%%%%%%%%%%%%%%%%%%%%%%%%%%%%%%%%%%
\section{Orienteering and Deadline TSP in Metrics of Bounded Treewidth}\label{sec:poly_p2p}

In this section, we prove Theorem \ref{thm:boundedTW} by first
developing a Dynamic Programming (DP) to solve {P2P orienteering} and {$k$-stroll} exactly on graphs with bounded treewidth. We then show how to extend this to obtain a QPTAS for deadline TSP on such graphs (when $\Delta$ is quasi polynomial).

Recall that in both {P2P orienteering} and {$k$-stroll}, we are given a graph $G$ with start and end nodes $s,t \in V$. In {P2P orienteering}, there is a length budget $B$, and the objective is to find an $s$-$t$ path of length at most $B$ that visits as many nodes as possible while in {$k$-stroll} we are given an integer $k$ and the goal is to find a minimum-length $s$-$t$ path visiting at least $k$ nodes. It is easy to verify that an exact algorithm for {$k$-stroll} also provides an exact algorithm for {P2P orienteering}: one needs to solve {$k$-stroll} for all possible values of $k$ and return the path with length at most $B$ for the largest $k$. However, it is important to note that an approximation for {$k$-stroll} does not imply an approximation for {P2P orienteering}. In particular, despite the fact that graphs of bounded doubling dimensions can be embedded into graphs with poly-logarithmic treewidth with $(1 + \epsilon)$-distortion (\cite{Talwar04}), an exact algorithm for {P2P orienteering} does not provide a $(1+\epsilon)$-approximation for the {P2P orienteering} in doubling dimension (so Theorem \ref{thm:P2PDBL} is not implied immediately). 
%However, the dynamic programming approach presented in the next section, while more complex than the current one, can be viewed as an extension where the solution is constructed on ``trees of clusters'' instead of a ``tree of bags''.

%%%%%%%%%%%%%%%%%%%%%%%%%%%%%%%%%%%%%%%%%%%%%%%%%%%%%%%%%%%%%%%%%%%
\subsection{A Polynomial Time Algorithm for {$k$-stroll} and {P2P Orienteering}}

Given $G=(V,E)$ with treewidth $\omega$, pairs of nodes $s,t\in V$, and $B\in\QQ^+$ is an instance of {P2P orienteering}. Let $P^*$ be an optimal solution for this instance with $k=|P^*|$. Note $P^*$ is a feasible solution for the instance of $k$-stroll on $G$ with specified target integer $k$ and the same start and end node $s,t\in V$. Suppose that $P'$ is an optimal solution for the $k$-stroll instance. Note $||P'||\leq ||P^*||\leq B$.

Let $T$ be a tree decomposition of $G$. In \cite{BodlaenderH95}, the authors showed that for a graph $G = (V, E)$ with treewidth $\omega$, one can create a binary decomposition tree $T$ with a width not more that $\omega'=3\omega+2=O(\omega)$ and a height at most $\rho\log n$ for some $\rho>0$. We shall assume that $T$ possesses these two additional properties and use $\omega$ (instead of $\omega'$) to refer to its treewidth.

Let $T$ be rooted at an arbitrary bag, denoted as $r$. For any bag $b$, we use $V_{b}$ to denote as the vertices in bag $b$, and $C_b$ as the union of vertices in bags below and including $b$. Let $G_b=G(C_b)$ represent the subgraph in $G$ over vertices in $C_b$. Observe that any path from a vertex in $C_b$ to a vertex in $V\setminus C_b$ must pass through $V_b$ (i.e. $b$ is a cut set).  We also use $b$ to refer to both the node in $T$ and its corresponding subset of vertices $V_b$ in $V$ if it is clear from the context. For any bag $b$, we define $R_b$ to be the set of edges with one endpoint in bag $b_1$ and the other in bag $b_2$, where $b_1$ and $b_2$ are children of $b$ in $T$.

We present a DP based on tree decomposition $T$ of $G$ that can find a path of length at most $||P'||$ visiting $k$ vertices.  For $v\in V$ let $T_{v}$ denote the bags in $T$ containing $v$. Note property 3 of $T$ implies for any vertex $v\in V$, $T_{v}$ is a connected subtree in $T$. In order to avoid over-counting, for every vertex $v\in T$, we consider placing a token on $v$ at the root of $T_{v}$, i.e. the bag containing $v$ that is closest to the root bag in $T$. Our goal is to find a shortest $s-t$-path in $G$ which picks up at least $k$ tokens.
Note ''a token in $T$ is picked up'' means that the corresponding vertex of the token in $G$ is visited. For a path $P$ we adapt the notation and use $|P|$ to refer to the number of tokens picked by $P$ in $T$.

Note that for a bag $b\in T$, $P'$ restricted to $G_{b}$ may be a collection of sub-paths where they all enter and exit
$G_{b}$ via $V_b$. Since $|b|=\omega+1$ and each such path continues down $b_1$ or $b_2$ (children of $b$), the number of such sub-paths is at most $O(\omega^2)$. Thus we can define a subproblem in the DP as an instance of multi-path {$k$-stroll} on $G_{b}$ with specified target integer $k_{b}$ and $\sigma_{b}\in O(\omega^2)$ pairs of 
$(s_{i},t_{i})$, $1\leq i\leq \sigma_b$, where the goal is to find $\sigma_{b}$ many $P_{i}$ paths, each path being from $s_{i}$ to $t_{i}$, in $G_{b}$ such that the total number of tokens picked up by $P_{1},\cdots,P_{\sigma_{b}}$ is at least $k_b$ and the total length of them $\sum_{i=1}^{\sigma_{b}}||P_{i}||$ is minimized. We use $A[b,k_{b},(s_{i},t_{i})_{i=1}^{\sigma_{b}}] $ to denote the subproblem defined above and the entry of the table stores the optimal value of this subproblem (i.e. the sum of the lengths of these paths).

We compute the entries of this DP working bottom-up on $T$.
The base cases are when $b$ is a leaf node of $T$. For such instances, the graph $G_b=G(V_b)$ has constant size and therefore each such subproblem can be easily solved by exhaustive search (explained in details soon).
For a non-leaf node $b\in T$, let $b_{1}$ and $b_{2}$ be the children bags of $b$ in $T$ and recall $R_{b}$ is the set of edges in $G_{b}$ with one endpoint in bag $b_1$ and the other in bag $b_2$. Note  $|R_{b}|\leq|V_{b_{1}}| |V_{b_{2}}|\leq (\omega+1)^{2} $.
We guess $k_{b_{1}}$ and $k_{b_{2}}$ for $b_{1}$ and $b_{2}$ such that $k_{b_{1}}+k_{b_{2}}=k_{b}$. Then we guess $\sigma_{b_{1}}$ pairs $\{ (s_{i},t_{i})\}_{i=1}^{\sigma_{b_{1}}} $ for $b_{1}$ and $\sigma_{b_{2}}$ pairs $\{ (s_{i},t_{i})\}_{i=1}^{\sigma_{b_{2}}} $ for $b_{2}$: first we guess the set of edges of $P_{1},\cdots,P_{\sigma_{b}}$  crossing between $b_{1}$ and $b_{2}$, which is a subset of $R_{b}$, denoted as $E_{b}$. For each edge in $E_{b}$ we guess it is in which one of the $\sigma_{b}$ path with start and end node pair $(s_{i},t_{i})$ and for each path with start and end node pair $(s_{i},t_{i}) $ we further guess the order of the guessed edges appearing. Specifically speaking, let $e_{1},e_{2},\cdots,e_{l}$ be the edges guessed in order appearing in the path  with start and end node pair $(s_{i},t_{i}) $. 
Without loss of generality, say $s_{i}\in V_{b_{1}}$ and $t_{i}\in V_{b_{2}}$. Then we set $s_{i}$ and the endpoint of $e_{1}$ in $ V_{b_{1}}$ to be a start and end node pair in $b_{1}$, the endpoint of $e_{1}$ in $ V_{b_{2}}$ and the endpoint of $e_{2}$ in $ V_{b_{2}}$  to be a start and end node pair in $b_{2}$, $\cdots$, the endpoint of $e_{l}$ in $ V_{b_{2}}$ and $t_{i}$ to be a 
 start and end node pair in $b_{2}$. By doing so we generate start and end node pairs
for $b_{1}$ and $b_{2}$ and we sort them based on their appearing in $s_{i}$-$t_{i}$ path and in the increasing order of $i$. This defines $\sigma_{b_{1}}$ and $\sigma_{b_{2}}$ start and end node pairs for $b_{1}$ and $b_{2}$. We formalize the recursion:

\begin{itemize}

\item for any bag $b\in T$  let $b_{1}$ and $b_{2}$ be the children bags of $b$.

\item  guess $k_{b_{1}}$ and $k_{b_{2}}$ for $b_{1}$ and $b_{2}$  such that $k_{b_{1}}+k_{b_{2}}  =k_{b}$.  

\item guess a subset $E_{b}\subset R_{b}$. For each edge in $E_{b}$, we guess it is in which of the $\sigma_{b}$ path with start and end node pair $(s_{i},t_{i})_{i=1}^{\sigma_{b}}$ and for each path with start and end node pair $(s_{i},t_{i})$ we guess the order of the edges appearing as described above. We generate $\{ (s_{i},t_{i})\}_{i=1}^{\sigma_{b_{1}}}$ and $\{ (s_{i},t_{i})\}_{i=1}^{\sigma_{b_{2}}}$ for $b_{1}$ and $b_{2}$. 

 \item $A[b,k_{b},(s_{i},t_{i})_{i=1}^{\sigma_{b}}]=\\\min_{k_{b_{1}}, k_{b_{2}},(s_{i},t_{i})_{i=1}^{\sigma_{b_{1}}},(s_{i},t_{i})_{i=1}^{\sigma_{b_{2}}} } A[b_{1},k_{1},(s_{i},t_{i})_{i=1}^{\sigma_{b_{1}}}] +A[b_{2},k_{2},(s_{i},t_{i})_{i=1}^{\sigma_{b_{2}}}]+
    \sum_{(u,v)\in E_{b}}d(u,v) $
    
\end{itemize}

The final answer we are seeking is $A[r, k, (s, t)]$  where $r$ is the root bag in $T$, $k$ and $s,t$ are specified in {$k$-stroll}. The base case is when $b$ is a leaf bag in $T$, i.e. $C_{b}$ is exactly $V_{b}$, thus $|C_{b}|\leq (\omega+1)$. Note $\sigma_{b}$ is at most $O(\omega^2)$ in this case because there are at most $(\omega+1)^{2}$ pairs of vertices in $G_{b}$. We can enumerate all possible collections of $P_{1},\cdots, P_{\sigma_{b}}$ such that $P_{i}$ is a $s_{i}$-$t_{i}$ paths. Specifically speaking, we guess all possible subset of $V_{b}$, which are $2^{\omega+1}$ many. Then for a specific subset, denoted as $U$, for each vertex in $U$ we guess it is in which one of the $\sigma_{b}$ path with start and end node pair $(s_{i},t_{i})$ and for each path with start and end node pair $(s_{i},t_{i})$ we guess the order of vertices of $U$ appearing on the path, which is at most $|U|!|U|^{\sigma_{b}}= (\omega+1)!(\omega+1)^{(\omega+1)^{2}}$ guessings. Among these enumeration of $P_{1},\cdots, P_{\sigma_{b}}$, which is at most $O(\omega^{\omega^2})$ many, we consider the one such that $|P_{1}\cup \cdots \cup P_{\sigma_{b}} |\geq k_{b}$ and  $\sum_{i=1}^{\sigma_{b}} ||P_{i}||$ is minimized.

We show the running time of computing one entry of the dynamic programming table is at most $O(\omega^{\omega^2}n)$:

In the recursion, for bag $b$ and its children $b_{1}$ and $b_{2}$, there are at most $k_b$ guesses for $k_{b_{1}}$ and $k_{b_{2}}$ such that $k_{b_{1}} +k_{b_{2}}=k_{b} $. For $E_{b}$:  because  $E_{b}\subset R_{b}$ and $|R_{b}|\leq (\omega+1)^{2} $ there are at most $O(2^{\omega^{2} })$ guesses. To generate $(s_{i},t_{i})$ ($1\leq i\leq\sigma_{b_{1}}$) for $b_{1}$ and $(s_{i},t_{i})$ ($1\leq i\leq \sigma_{b_{2}}$) for $b_{2}$, for a certain $E_{b}$ and for each edge in $E_{b}$ we guess this edge is in which one of $\sigma_{b}$ paths with start and end node pair $(s_{i},t_{i})$ ($1\leq i\leq\sigma_{b}$) and for each path with source-sink pair $(s_{i},t_{i})$ we guess the order of the edges appearing in it, which are at most $|E_{b}|! |E_{b}|^{\sigma_{b}} $ guesses. Note a start and end node pair in $\sigma_{b}$ is a pair of vertices in $V_{b}\cup (\{s,t\}\cap C_{b}) $, thus $\sigma_{b}\leq O(\omega^2)$ and therefore  $|E_{b}|! |E_{b}|^{\sigma_{b}}$ is at most  $O((\omega^2)!\omega^{\omega^{2}})$.

Now we bound the size of the DP table by $O(n^{\rho+1}\omega^{\omega^2})$:
Recall an entry of the table is $A[b,k_{b},(s_{i},t_{i})_{i=1}^{\sigma_{b}}]$. 
Since $T$ is binary and of height $\rho\log n$, there are $O(n^\rho)$ many bags.
There are at most $n$ choices for $k_b$ and $O(\omega^2)$ choices for
$(s_{i},t_{i})$, $1\leq i\leq\sigma_b$.

\begin{theorem}
    Let $G=(V,E)$ be a graph with tree width $w$, source-sink nodes $s,t\in V$, and integer $k$ as an instance of {$k$-stroll}. We can find an optimum solution in time $O(n^{\rho+2}\omega^{\omega^2})$ time. 
\end{theorem}

For the P2P orienteering instance on $G$ with specified budget $B$, source-sink nodes $s,t\in V$, we try all possible $k$ (from $1$ to $n$) for the size of the optimum path $P^*$ and for each $k$
we compute the optimal for the $k$-stroll instance on $G$ with the corresponding $k$. We return the maximum $k$ such that the length of the optimal for corresponding $k$-stroll instance on $G$ is at most $B$. This together with the above theorem complete
the proof of the first part of Theorem \ref{thm:boundedTW}.

%%%%%%%%%%%%%%%%%%%%%%%%%%%%%%%%%%%%%%%%%%%%%%%%%%%%%%%%%%%%%%%%
\subsection{An Approximation Scheme for Deadline TSP}

In this section we show how we can extend the exact algorithm for {P2P orienteering} on bounded treewidth graphs to an approximation scheme for deadline TSP assuming that distances and deadlines are integer, which completes the proof of Theorem \ref{thm:boundedTW}. The idea is to guess a sequence of vertices of the optimum where the $\mu$-excess is increasing geometrically (say by a factor of $(1+\eps)^i$), then try to solve {P2P orienteering} concurrently for the instances defined between those vertices and then drop a small fraction of the points from the solutions for the {P2P orienteering} instances so that the saving in time is enough to make sure all the vertices are visited by their deadlines.
This is building upon ideas of \cite{FriggstadS17} and ideas we develop later for approximation of {P2P orienteering} on doubling metrics (Section \ref{sec:qptas_deadlineTSP}).

Let $G=(V,E)$ be a graph with bounded treewidth $\omega$, a start node $s\in V$ and $D(v)$ for all $v\in V$ as an instance of deadline TSP on $G$. Let $n=|V|$ and $\delta=\log \Delta_{G}$.
Let $\mu=\lfloor \frac{1}{\eps} \rfloor+1$ and $\alpha=(1+\eps)$.  Let $P^{*}$ be an optimal for this deadline TSP instance and $\langle v_{0},v_{1},\cdots,v_{m}\rangle$ be a sequence of vertices in $P^{*}$ satisfying the following properties:

\begin{itemize}
    \item $v_{0}=s$ is the start node of $P^{*}$.
    \item  $v_{i+1}$ is the first vertex in $P^{*}$ after $v_{i}$ such that $ \calE_{P^{*}_{v_{i}v_{i+1}},2}> \alpha^{i}$, except possibly for $v_{m}$ which is the last vertex of $P^{*}$.
\end{itemize}

We also denote the vertex on $P^*$ just before $v_{i+1}$ by $v'_i$. Since $||P^*||\leq n\Delta$, $m\leq h\delta$ for some $h=h(\eps)>0$.
We can assume that $|P^*_{v_iv_{i+1}}|\geq \mu^2$, otherwise we can compute $P^*_{v_iv_{i+1}}$ exactly using exhaustive search.
For each $0\leq i <m$, we break $P^*_{v_iv_{i+1}}$ into $\mu-1$ subpaths of (almost) equal sizes, denoted as  $P^*_{i,j}$, for $1\leq j <\mu$,
by selecting a $\mu$-jump $J_i: v_i=u^1_i,u^2_i,\ldots,u^{\mu}_i=v_{i+1}$ of $P^*_{v_iv_{i+1}}$ and letting $P^*_{i,j}$ be the segment of $P^*$ between $u^j_i$ and $u^{j+1}_i$. The $\mu$-jump $J_i$ is defined as follows. Assume $P^*_{v_iv_{i+1}}$ has size $k_{i}$ and say $P^*_{v_iv_{i+1}} =\langle v_{i,1},\cdots,v_{i,k_i} \rangle $ where $v_i=v_{i,1}$ and $v_{i,k_i}=v_{i+1}$. In other words, if we let $a_{j}= \lceil \frac {(j-1)(k_i-1)}{\mu -1} \rceil +1$ then $a_{1}=1$, $a_{\mu}=k_{i}$, and if 
we consider $v_{i,a_1},\ldots,v_{i,a_{\mu}}$ then we obtain $J_i$ by letting
$v_{i,a_j}=u^j_i$. 
Suppose $J^*_\mu=J^*_\mu(P^*_{v_iv_{i+1}})$ is the optimum $\mu$-jump of $P^*_{v_iv_{i+1}}$, which is the $\mu$-jump with the maximum length. Recall that $\calE_{P^*_{v_iv_{i+1}},\mu}$ is the $\mu$-excess of $P^*_{v_iv_{i+1}}$ and with  
$B_i= ||J^*_\mu(P^*_{v_iv_{i+1}})||+\calE_{P^*_{v_iv_{i+1}},\mu}$ we have $||P^*_{v_iv_{i+1}}||= B_i$. To simplify the notation we denote the $\mu$-excess of subpath $P^*_{v_iv_{i+1}}$ by $\calE^*_{i,\mu}$. We also use $\calE'_{i,\mu}$ to denote the $\mu$-excess of path $P^*_{v_iv'_i}$. 
Let $v$ be an arbitrary vertex in $P^*_{v_iv_{i+1}}$ that falls in between 
$u^j_i$ and $u^{j+1}_i$. We use $||J_i(v_i,u^{j}_i)||$ to denote the length of the $J_i$ path from $v_i$ to $u^j_i$ (i.e. following along $J_i$ from the start node $v_i$ to $u^j_i$). Define $L_{i,j}=\sum_{j=0}^{i-1} ||P^*_{v_j v_{j+1}}||+||J_i(v_i,u^{j}_i)||$. 
Note that the visiting time of $v$ in $P^*$ (and hence the deadline of $v$) is lower bounded by $L_{i,j}$.

Let $N_{i,j}=\{ v: D(v)\geq L_{i,j} \}$. Note $P^{*}_{i,j}$ ($0\leq i<m, 1\leq j<\mu$) is a feasible solution of the multiple groups-legs orienteering instance (Definition \ref{def:mglo}) with groups $ 0\leq i<m$ and legs $1\leq j<\mu$, start and end node pairs $(u_i^{j}, u_i^{j+1})$, budgets $B_{i}$ and subset $N_{i,j}$. We will show there is a $(1+\epsilon)$-approximation i.e. a set of paths $Q_{i,j}$ such that $Q_{i,j}$ is a  $u^j_{i},u^{j+1}_i$-path and if we define concatenation of different legs of group $i$ by $Q_i=Q_{i,1}+\ldots+Q_{i,\mu-1}$ then 
\begin{equation}\label{eqn:bteq1}
    ||Q_i||=\sum_{j=1}^{\mu-1}||Q_{i,j}||\leq ||P^*_{v_iv_{i+1}}||-\eps\calE^*_{i,\mu}\quad\quad\text{and}\quad\quad
|\bigcup_{i=0}^m Q_i|=|\bigcup_{i=0}^m\cup_{j=1}^{\mu-1}(Q_{i,j}\cap N_{i,j})|\geq (1-3\eps)|P^*|.
\end{equation} 
We also show that if $v$ is visited by $Q_{i,j}$ then if $Q_{i,j}(u^j_i,v)$ denotes the segment of path $Q_{i,j}$ from $u^j_i$ to $v$, then the length of the segment from $v_i$ to $v$ in $Q_i$ can be upper bounded:
\begin{equation}\label{eqn:bteq2}
||Q_i(v_i,v)||=\sum_{\ell=1}^{j-1}||Q_{i,\ell}|| + ||Q_{i,j}(u^j_i,v)|| \leq ||J_i(v_i,u^j_i)||+(1-\eps)\calE'_{i,\mu}.
\end{equation}

We will show the existence of such paths $Q_{i,j}$  and also how to find these using a DP.
For now suppose we have found such paths $Q_i$ as described above. We concatenate all these paths to obtain the final answer $\calQ=Q_0+Q_1+\ldots+Q_{m}$. We show the vertices of $\calQ$ 
are visited before their deadlines and hence we have an approximation  for deadline TSP. Given the bounds given for the sizes of $Q_i$'s in (\ref{eqn:bteq1}),  the number of vertices visited overall (respecting their deadlines) is at least $(1-3\eps)|P^*|$.

To see why the vertices in $\calQ$ respect their deadlines consider an arbitrary node $v\in Q_i$.
Note that each $Q_i$ contains the vertices in $J_i$ (as those are the vertices that define $\mu-1$ legs of the $i$'th group). Suppose $v$ is visited in $Q_{i,j}$, i.e. between $u^j_i$ and $u^{j+1}_i$. Therefore, the visit time of $v$ in $\calQ$, i.e. $||\calQ_{sv}||$ is bounded by:

\begin{eqnarray*}
  ||\calQ_{sv}|| &=& \sum_{\ell=0}^{i-1}||Q_\ell|| + ||Q_i(v_i,v)|| \\
    &\leq& \sum_{\ell=0}^{i-1} (||P^*_{v_\ell v_{\ell+1}}|| - \eps\calE^*_{\ell,\mu}) + ||J_i(v_i,u^j_i)||+(1-\eps)\calE'_{i,\mu} \quad\quad\quad\quad\quad\mbox{using (\ref{eqn:bteq1}) and (\ref{eqn:bteq2})}\\
    &=& L_{i,j} + (1-\eps)\calE'_{i,\mu}-\eps\sum_{\ell=0}^{i-1}\calE^*_{\ell,\mu}\\
    &\leq& D(v)
\end{eqnarray*}
where the last inequality follows from the fact that $\calE'_{i,\mu}\leq\alpha^i-1$ and
$\calE^*_{\ell,\mu}>\alpha^\ell$ so $\eps\sum_{\ell=0}^{i-1}\calE^*_{\ell,\mu}\geq  \eps\sum_{\ell=0}^{i-1}\alpha^\ell = (\alpha^i-1)\geq \calE'_{i,\mu}$.

So we need to show the existence of $Q_i, 0\leq i<m$ as described and how to find them. To simplify notation, for each $i$, let $P^{*}_{i}$ be the subpath $P^*_{v_iv_{i+1}}$ and $P^{*}_{i,j} $ be the subpath $P^*_{u_i^j u_{i}^{j+1}}$. We start from $P^{*}_{i,j}$ ($1\leq j<\mu$), for each $P^{*}_{i,j}$ we consider the $2$-excess of it and let $j'$ be the index that $P^*_{i,j'}$ has the largest  2-excess among indices $1,\ldots,\mu-2$. We consider short-cutting the two subpaths $P^*_{i,j'}$ and $P^*_{i,\mu-1}$ and let the resulting path obtained from $P^*_i$ by these short-cutting be $Q_i$. In other words, $Q_i$ is the same as $P^*_i$ except that each of $P^*_{i,j'}$ and $P^*_{i,\mu-1}$ are replaced with the direct edges $(u^{j'}_i,u^{j'+1}_i)$ and $(u^{\mu-1}_i,v'_i)$, respectively.
Let $D^2_i=\calE_{P^*_{i,j'},2}+\calE_{P^*_{i,\mu-1},2} $. We have $\calE_{P^*_{i,j'},2}  \geq\frac{1}{\mu-1}\sum_{j=1}^{\mu-2}\calE_{P^*_{i,j},2}$. Recall that $v'_i$ is the vertex on $P^*$ just before $v_{i+1}$, so $||P^*_i(v_i,v'_i)||=||J^*_\mu(P^*_i(v_i,v'_i))||+\calE'_{i,\mu}$ (this is from the definition of optimum $\mu$-jump of the path $P^*_i(v_i,v'_i)$). So for vertices in $Q_{i}$ we can bound the length of the subpath from $v_i$ to $v$ by:  
\begin{equation}\label{}
||Q_i(v_i,v)||\leq ||P^*_i(v_i,v)||- \calE_{P^*_{i,j'},2} \leq ||J_\mu(P^*_i(v_i,u_i^j))||+\calE'_{i,\mu}-\frac{1}{\mu-1}\sum_{j=1}^{\mu-2}\calE_{P^*_{i,j},2}\leq ||J_i(v_i,u_i^j)|| +(1-\eps)\calE'_{i,\mu}.
\end{equation}

Also for the total length of $Q_i$ we have:
\begin{equation}\label{}
    ||Q_i||\leq ||P^*_i|| - D^2_i \leq ||P^*_i||-\frac{1}{\mu-1}\sum_{j=1}^{\mu-1}\calE_{P^*_{i,j},2}\leq ||P^*_{v_{i}v_{i+1}}||-\eps \calE^{*}_{i,\mu} .
\end{equation}

As for the size note that $|P^{*}_{i,j'}|=a_{j'+1}-a_{j'}+1 =(\lceil \frac {(j'+1)(k_{i}-1)  }{\mu -1} \rceil +1) -  (\lceil \frac{j'(k_{i}-1)  }{\mu-1}  \rceil+1) +  1 \leq \lceil\frac{k_{i}-1}{\mu-1}  \rceil +1 $. Same bound holds for $|P^*_{i,\mu-1}|$.
From the construction of $Q_{i}$: $|Q_{i}|=k_{i}-|P^{*}_{i,j'}|+2  -|P^*_{i,\mu-1}| +2 \geq k_{i} - 2\lceil\frac{k_{i}-1}{\mu-1}  \rceil +4 \geq  k_{i} - 2\lfloor\frac{k_{i}-1}{\mu-1} \rfloor  \geq (1-3\epsilon) k_{i} $, since we assumed $k_i\geq\mu^2$.

Now we describe our approximation algorithm for deadline TSP on the given
instance $\calI$ based on graph $G$ of bounded $\omega$ treewidth.
The algorithm has two phases. In Phase 1 we guess the vertices
$v_1,\ldots,v_m$ of optimum (as described) as well as $u^1_i,\ldots,u^{\mu-2}_i$ for each $0\leq i<m$, and also guess $\calE^*_{i,\mu}$ and $||J^*_\mu(P^*_{v_iv_{i+1}})||$; so we have
guessed $B_i=||J^*_\mu(P^*_{v_iv_{i+1}})||+\calE^*_{i,\mu}$, which imply sets $N_{i,j}$ as well. For each $i$ we can guess $k_i=|P^*_{v_iv_{i+1}}|$ and for those $k_i<\mu^2$ we guess $P^*_{v_iv_{i+1}}$ exactly. All the guesses in Phase 1 can be done in time $(n\Delta)^{O(\mu^2h\delta)}$. 

In Phase 2 we find near optimum solution $Q$ by finding $Q_i$'s as described for those $i$'s for which $k_i\geq\mu^2$ using DP. In the remaining we assume all $i$'s satisfy $k_i\geq\mu^2$.
Let $T$ be a tree decomposition of $G$ that it is binary with height $ \rho\log n$ for some constant $\rho>0$ and the width of $T$ is at most $\omega$. 
%For any $0\leq i<m$ where $|P^*_{v_iv_{i+1}|\leq$
We present a DP based on the tree decomposition $T$ that can find $Q_{0},Q_{1},\cdots,Q_{m}$. Recall for any vertex $v\in V$, the bags in $T$ containing $v$, i.e. $T_{v}$ is a connected subtree in $T$. In order to avoid over-counting, for every vertex $v\in T$, we consider placing a token on $v$ at the root of $T_{v}$. We adapt the notation and use $|Q_{i,j}\cap N_{i,j}|$ to refer to the number of tokens picked by $Q_{i,j}$ in $N_{i,j}$. 

Note for bag $b$, let $C_{b}$ denotes the union of vertices in bags below and including $b$ and $G_{b}$ represent the subgraph in $G$ over vertices in $C_{b}$. For any bag $b\in T$, any group 
$0 \leq i <m$, and any leg $1 \leq j <\mu$, the restriction of $Q_{i,j}$ in the subgraph $G_{b}$ may be a collection of sub-paths where they all enter and exit $G_{b}$ via $V_{b}$ where the number of such sub-paths is at most $O(\omega^{2})$ because $|V_{b}|\leq \omega+1$. We can define a subproblem in the DP as an instance of multi-groups-legs multi-paths orienteering on $G_{b}$ with groups $0\leq i<m$ and legs $1\leq j<\mu$  start and end node pairs  $(s_{i,j,l},t_{i,j,l}),1\leq l \leq \sigma_{b,i,j}$, budgets $B_{b,i}$ and subset $N_{i,j}\subset V$. The goal is to find a collection of path $ Q_{i,j,l}$ such that $Q_{i,j,l} $ a $s_{i,j,l}$-$t_{i,j,l}$ path,  $\sum_{j=1}^{\mu-1}\sum_{l=1}^{\sigma_{b,i,j}}||Q_{i,j} ||$ is at most $B_{b,i}$ and $| \cup_{i=0}^{m-1} \cup_{j=1}^{\mu-1} ( \cup_{l=1}^{\sigma_{b,i,j}} Q_{i,j,l})\cap N_{i,j}  )  |  $ is  maximized. We use $A[b, \{ B_{b,i}\}_{0\leq i<m}, \{ (s_{i,j,l},t_{i,j,l})_{l=1}^{\sigma_{b,i,j}}\}_{0\leq i<m;1\leq j<\mu}]$ to denote the subproblem defined above and entry of the table store the optimal value of the subproblem.

We compute the entries of this DP from bottom to up of $T$. The base cases are when $b$ is a leaf bag of $T$, where $G_{b}$ has constant size hence each such subproblem can be solved by exhaustive search. In the recursion, consider any entry $A[b, \{ B_{b,i}\}_{0\leq i<m}, \{ (s_{i,j,l},t_{i,j,l})_{l=1}^{\sigma_{b,i,j}}\}_{0\leq i<m;1\leq j<\mu}] $, let $b_{1}$ and $b_{2}$ be the children bags of $b$ and let $ R_{b}$ is the set of edges in $G_{b}$ with one endpoint in $b_{1}$ and the other in $b_{2}$. Note $|R_{b}|\leq |V_{b_{1}}| |V_{b_{2}}| = O(\omega^{2})$. For each group $i$ and each leg $j$, first we guess a subset of $R_{b}$ (the set of edges $Q_{i,j,l}, 1\leq l\leq \sigma_{b,i,j}$ crossing between $b_{1}$ and $b_{2}$) such that they are disjoint (for different $i,j$) and for every edge in $E^{i,j}_{b}$ both end points are in $N_{i,j}$ and , denoted as $E^{i,j}_{b}$. For each $i$ we guess $B_{b_{1},i}$ and $B_{b_{2},i}$ such that  $B_{b_{1},i}+B_{b_{2},i} +\sum_{j=1}^{\mu-1} \sum_{(u,v) \in E^{i,j}_{b} }d(u,v)=B_{b,i} $. 
We show how to guess $(s_{i,j,l},t_{i,j,l})_{l=1}^{\sigma_{b_{1},i,j}}$ for $b_{1}$ and $(s_{i,j,l},t_{i,j,l})_{l=1}^{\sigma_{b_{2},i,j}}$ for $b_{2}$ and check the consistency of them: for each $E^{i,j}_{b} $ and for each edge in $E^{i,j}_{b} $ we guess it is in which one of the $\sigma_{b,i,j}$ path with the start and end node pair $(s_{i,j,l},t_{i,j,l})_{l=1}^{\sigma_{b,i,j}}$ and for each path with start and end node pair $(s_{i,j,l},t_{i,j,l})$ we guess the order of the guessed edges appearing on the path. Specifically speaking, let $e_{1},e_{2}\cdots,e_{w}$ be the edges guessed in the path with start and end node pair $(s_{i,j,l},t_{i,j,l})$ appearing in this order. Without loss of generality, say $s_{i,j,l} \in V_{b_{1}} $ and $t_{i,j,l} \in V_{b_{2}} $.  Then we set $s_{i,j,l}$ and the endpoint of $e_{1}$ in $V_{b_{1}}$ to be a start and end node pair in group $i$ and leg $j$ in $b_{1}$, the endpoint of $e_{1}$ in $V_{b_{2}}$  and the  endpoint of $e_{2}$ in $V_{b_{2}}$ to be a start and end node pair in group $i$ and leg $j$ in $b_{2}$, $\cdots$, the endpoint of $e_{w}$ in $V_{b_{2}}$ and $t_{i,j,l}$ to be a start and end node pair in group $i$ and leg $j$ in $b_{2}$. By doing so we generate start and end node pairs in group $i$ and leg $j$ for $b_{1}$ and $b_{2}$ and we sort them based on their appearing in $s_{i,j,l}$-$t_{i,j,l}$ path. This defines $\sigma_{b_{1},i,j}$ and $\sigma_{b_{2},i,j} $ start and end node pairs for $b_{1}$ and $b_{2}$. Formally, to compute $A[b, \{ B_{b,i}\}_{0\leq i<m}, \{ (s_{i,j,l},t_{i,j,l})_{l=1}^{\sigma_{b,i,j}}\}_{0\leq i<m;1\leq j<\mu}] $:
\begin{itemize}

\item  let $b_{1}$ and $b_{2}$ be the children bags of $b$. Let $R_{b} $ be the set of edges in $G_{b}$ crossing $b_{1}$ and $b_{2}$. 

\item for each $i$ and each $j$, we guess a subset of $ R_{b}$, denoted as $E^{i,j}_{b}$ such that they are disjoint (for different $i,j$) and for every edge in $E^{i,j}_{b} $ both end points are in $N_{i,j}$.

\item for each $i$, we guess $B_{b_{1},i}$ and $B_{b_{2},i}$ such that $ B_{b_{1},i}+B_{b_{2},i}  + \sum_{j=1}^{\mu-1}\sum_{(u,v)\in E^{i,j}_{b}} d(u,v) =B_{b,i} $. 

\item for each $E^{i,j}_{b} $ and each edge in $E^{i,j}_{b}$, we guess it is in which one of the $\sigma_{b,i,j}$ path with start and end node pair $(s_{i,j,l},t_{i,j,l})_{l=1}^{\sigma_{b,i,j}}$ and for each path with source-sink pair $(s_{i,j,l},t_{i,j,l})$ we guess the order of the edges appearing on the path as described above. We generate start-end pairs in group $i$ and leg $j$  for $b_{1}$ and $b_{2}$ accordingly. Then:

\item $A[b,\{ B_{b,i}\}_{0\leq i<m}, \{ (s_{i,j,l},t_{i,j,l})_{l=1}^{\sigma_{b,i,j}}\}_{0\leq i<m;1\leq j<\mu}] =\\ \max [b_1, \{ B_{b_1,i}\}_{0\leq i<m}, \{ (s_{i,j,l},t_{i,j,l})_{l=1}^{\sigma_{b_1,i,j}}\}_{0\leq i<m;1\leq j<\mu}] + [b_2, \{ B_{b_2,i}\}_{0\leq i<m}, \{ (s_{i,j,l},t_{i,j,l})_{l=1}^{\sigma_{b_2,i,j}}\}_{0\leq i<m;1\leq j<\mu}]  $, where the maximum is taken over all tuples\\ $( \{B_{b_1,i}\}_{0\leq i<m}, \{B_{b_2,i}\}_{0\leq i<m}, \{ (s_{i,j,l},t_{i,j,l})_{l=1}^{\sigma_{b_1,i,j}}\}_{0\leq i<m;1\leq j<\mu}, \{ (s_{i,j,l},t_{i,j,l})_{l=1}^{\sigma_{b_2,i,j}}\}_{0\leq i<m;1\leq j<\mu})$ as described above.

\end{itemize}

As said earlier, we have guessed
$v_0,v_1,\ldots,v_m$, $u^1_i,\ldots,u^{\mu-2}_i$, and $B_i$'s in Phase 1. The goal is to compute $A[r,\{ B_{i}-\eps \calE^{*}_{i,\mu} \}_{0\leq i<m},\{(u_i^{j},u_i^{j+1})\}_{0\leq i<m;1\leq j<\mu} ]$ where $r$ is the root bag in $T$.

Now we analyze the running time of DP. First we show the running time of computing one entry of the dynamic programming table is at most $n^{O((\frac{\delta}{\eps})^{2})}$. In the recursion, for bag $b$ and for $\{E^{i,j}_{b}\}_{0\leq i<m;1\leq j<\mu}$: for each $i$ and leg $j$, because $E^{i,j}_{b} \subset R_{b}$ and $|R_{b}|\leq (\omega+1)^{2}$ thus there are at most $ [2^{(\omega+1)^{2}}]^{m\mu}$ many possible $E^{i,j}_{b} $ to consider for all $i$ and all $j$. There are at most $ (n\Delta_{G})^{m\mu}=n^{O((\frac{\delta}{\eps})^{2})}$ guessings for $B_{b_{1},i}$ and $ B_{b_{2},i}$ for $b_{1}$ and $b_{2}$ such that $ B_{b_{1},i}+B_{b_{2},i}  + \sum_{j=1}^{\mu-1}\sum_{(u,v)\in E^{i,j}_{b}} d(u,v) =B_{b,i} $ for all $i$. To generate $ \{(s_{i,j,l},t_{i,j,l})_{ l=1}^{\sigma_{b_{1},i,j}}\}_{0\leq i<m; 1\leq j<\mu}$ and  $\{(s_{i,j,l},t_{i,j,l})_{ l=1}^{\sigma_{b_{2},i,j}}\}_{0\leq i<m; 1\leq j<\mu}$ for $b_1$ and $b_{2}$: for each group $i$ and leg $j$ and for each edge in $E^{i,j}_{b}$ we guess it is in which one of $\sigma_{b,i,j}$ path with start and end node pair $(s_{i,j,l},t_{i,j,l})$, $1\leq l\leq\sigma_{b,i,j}$, and for each path with start and end node pair $(s_{i,j,l},t_{i,j,l})$ we guess the order of the edges appearing, which is at most $ |E^{i,j}_{b}|! |E^{i,j}_{b}|^{\sigma_{b,i,j}}$ guessings. Note a start and end node pair in $\sigma_{b,i,j}$ is a pair of vertices in $V_{b}\cup(\{u_i^j, u_i^{j+1}\}\cap C_{b})$, thus $\sigma_{b,i,j}\leq (\omega+3)^{2}= O(\omega^{2})$. Therefore the total guessings for all $i$ and all $j$ is at most $ (|E^{i,j}_{b}|! |E^{i,j}_{b}|^{\sigma_{b,i,j}})^{m\mu}\leq ((\omega+1)^{2}!(\omega+1)^{2 (\omega+3)^{2}})^{(m+1)\mu}=n^{O(\omega^2\delta/\eps)})$.

We show the size of the dynamic programming table is at most $n^{O((\frac{\omega\delta}{\eps})^{2})}$. Recall an entry of the table is $A[b,\{ B_{b,i}\}_{0\leq i<m}, \{ (s_{i,j,l},t_{i,j,l})_{l=1}^{\sigma_{b,i,j}}\}_{0\leq i<m;1\leq j<\mu}] $. For $b$, since $T$ is binary and the height of $T$ is $\rho \log n$, there are $O(n^{\rho})$ many bags. For $B_{b,0},\cdots,B_{b,m} $, there are at most $(n\Delta_{G})^{m}$ choices for and for $\{ (s_{i,j,l},t_{i,j,l})_{l=1}^{\sigma_{b,i,j}}\}_{0\leq i<m;1\leq j<\mu}$, there are at most $ ((\omega+1)^{2(\omega+3)^{2}})^{m\mu}$ possible start-end node pairs to consider.
Combined the overall time for Phase 1 and Phase 2 to compute a $(1+\eps)$-approximation of $P^*$ for deadline TSP is $n^{O((\omega\delta/\eps)^2)}$. This completes the proof of the 2nd part of Theorem \ref{thm:boundedTW}.

%%%%%%%%%%%%%%%%%%%%%%%%%%%%%%%%%%%%%%%%%%%%%%%%%%%%%%%%%%%%%%%%%%
\section{An Approximation Scheme for {$k$-stroll} and Orienteering in Doubling Metrics}\label{sec:qptas_p2p}

In this Section we prove Theorems \ref{thm:kstrollDBL} and \ref{thm:P2PDBL}.
Suppose that $G=(V,E)$ (a graph with bounded doubling dimension $\kappa$), $s,t\in V$, and budget $B\in \QQ^+$ are given as an instance of {P2P orienteering}. Let $\delta=\log \Delta_{G}$. We assume $\Delta_{G}$ is quasi-polynomial  in which case $\delta$ is polylogarithmic. 

{\bf Overview:}
Before we present our approximation scheme for {P2P orienteering} on $G$ let us review the algorithm of Chen and Har-Peled \cite{chen2008euclidean} for {P2P orienteering} on Euclidean plane. They showed that the algorithm of Mitchell \cite{mitchell1999guillotine} for Euclidean TSP in fact implies a $(\epsilon,u)$-approximation for {$k$-stroll}, i.e. an algorithm that finds a path $P$ with $|P|=k$ where $||P||\leq ||P^*||+\eps\cdot\calE_{P^*,\mu}$, where $P^*$ is the optimum {$k$-stroll} solution. 
Using this approximation algorithm to solve the orienteering problem (with the assumption that $k$ is the guessed value of the optimum solution for orienteering), they show that one can break the path obtained by the algorithm into $u=O(\frac{1}{\eps})$ segments, each containing $\eps\cdot k$ vertices; this will be a $\mu$-jump; once we remove one segment that has the longest length the total length of the path is dropped by at least $\eps\cdot\calE_{P^*,\mu}$, and therefore the length of the resulting path is at most the length of the optimum path $||P^*||$ and hence below
the given budget and the total number of nodes is at least $(1-\epsilon)k$. This yields the approximation scheme for orienteering (via {$k$-stroll}).
The algorithm for $(\epsilon,u)$-approximation of {$k$-stroll} is essentially the algorithm of Mitchell \cite{mitchell1999guillotine} for Euclidean TSP which breaks the problem into (polynomially many) subinstances, each defined by a window $w$ (a minimum bounding box). Mitchell defines a vertical (or horizontal) line $\ell$ that cuts a windows $w$ a {\em cut} and  for a parameter $m=O(1/\epsilon)$ if the number of edges of the optimum path that cross $\ell$ is no more than $m$ then this cut is {\em sparse}; else $\ell$ is dense. If $\ell$ is sparse one can {\em guess} the edges of the optimum in time $O(n^m)$ and break the problem into independent instances. If there is no sparse cut then the value of the optimum in this window is large. In this case using either Mitchell's \cite{mitchell1999guillotine} scheme (using bridges) or Arora's scheme \cite{Arora98} for TSP (making paths portal respecting paths) one can modify the solution for the problem restricted to $w$ to a near optimum one of length at most $(1+\frac{1}{m})||P^*(w)||$, where $P^*(w)$ is the restriction of the optimum path to window $w$. We should point out that this idea of partitioning sub-instances into sparse and dense regions has been used in other works (e.g. in \cite{BartalGK16} to obtain a PTAS for TSP on doubling metrics). 

Our QPTAS for orienteering on doubling metrics builds upon this idea of sparse and dense sub-instances, but the difficulty is we do not have a polynomial size set of windows or cuts as in the Euclidean case. 
We first try to find a good approximation for {$k$-stroll} on $G$ with a similar stronger upper bound on the quality of the solution:
a solution with a stronger upper bound of the form $||P^*||+\epsilon\cdot\calE_{P^*,u}$. We show the existence of a near-optimum solution for {$k$-stroll} which has some structural properties. Suppose that $T$ is a hierarchical decomposition of $G$ as in Theorem \ref{thm:hier-dec}. Consider an arbitrary node $C\in T$ of the hierarchical decomposition at level $i$ with children $C_1,\ldots,C_\sigma$. 
%We denote the graph induced by the vertices of $G$ corresponding to the node $c$ (or $c_i$) of $T$ by $G_c$ (or $G_{c_i}$). 
Roughly speaking we  would like this near optimum solution to have very few edges crossing between different clusters $C_i$'s when we focus on $C$ or these edges be portal-respecting if we impose some portals for each $C_i$ as in the algorithm for TSP on doubling metrics \cite{Talwar04}. The existanse of such a structured near optimum allows us to find it using a DP on $T$.

For simplicity of presentation, suppose $P^*$ (for {$k$-stroll}) is known (let's call this "Assumption 1"). As said above, consider $C\in T$ with children $C_1,\ldots,C_\sigma$.
We also consider the restriction of $P^*$ to $C$, denoted by $P^*_C=P^*\cap C$ and let $\Delta_C$ denote the diameter of $C$. We consider two cases based on whether $||P^*_C||\leq \frac{\log n}{\eps}\Delta_C$ holds or not. If the inequality holds we call $P^*_C$ {\em sparse} with respect to $C$; else we call it {\em dense}.

{\bf Sparse case:} If $||P^*_C||\leq \frac{\log n}{\eps}\Delta_C$, we show in this case the {\em expected} number of edges of $P^*_C$ crossing between different subgraphs $C_1,\ldots,C_\sigma$ is $O(\frac{\kappa\log n}{\epsilon})$; this uses property 4 of Theorem \ref{thm:hier-dec}. For simplicity assume this actually holds with high probability instead of "in expectation" (call this "Assumption 2")

{\bf Dense case:}
If $||P^*_C||>\frac{\log n}{\epsilon} \Delta_C$ we modify $P^*_C$ to a near optimum one by making $P^*_C$ to be portal respecting when it crosses between $C_1,\ldots,C_\sigma$ (as was done for TSP for e.g. in \cite{Talwar04}). We consider each of $C_i$
and generate a set of $O(\log^\kappa n)$ portals for it. We make $P^*_C$ portal respecting, that is, we modify it such that it crosses between different subgraphs $C_i$ only through portals to reduce the number of times it crosses between different subgraphs to be poly-logarithmic. We show the {\em expected} increase of length of $P^*_C$ in $C$ is bounded by $O(||P^*_C||\cdot\eps/\log n)$, which holds by Theorem \ref{thm:hier-dec}. Again, assume that this actually holds with high probability instead of "in expectation" every time we do this (call this "Assumption 3")

We modify $P^*$ to a structured near optimum solution $P'$ as follows. Starting from the root of $T$, and let initially $P'=P^*$. We modify $P'$ as we go down the tree $T$: for each node $C\in T$, 
if $P'_C$ (restriction of $P'$ to $C$) is sparse w.r.t. $C$ we make no modification going down to children of $C$. However, if $P'_C$ is dense w.r.t. $C$, we make $P'_C$ portal respecting when crossing between the subgraphs corresponding to children of $C$. This increases the length of $P'_C$ by at most $O(||P'_C||\cdot\eps/\log n)$. Given that height of $T$ is $O(\log n)$, the total increase in length of $P'$ over all steps can be bounded to be $O(\eps||P^*||)$.

We have made a number of assumptions above that we should remove.
We don't know $P^*$ (Assumption 1) and whether it is sparse or dense in each step we are going down the tree $T$. To remove this assumption we {\em guess} both cases of whether $P^*$ is sparse or dense at each step and add both into the DP.
To remove Assumptions 2 and 3, we repeat the random decomposition of Theorem \ref{thm:hier-dec} $\Omega(\log n)$ many times at each step. In other words, imagine we do parallel (or non-deterministic) runs of the decomposition at each step we want to partition a cluster. Then instead of things holding in expectation, one can show that in at least one such decomposition Assumptions 2 and 3 actually hold with high probability. This idea of repeating random decompositions has been used in earlier works, most recently by \cite{Cohen-Addad22} to present approximation scheme for $k$-MST on minor free graphs.
This non-deterministic decomposition can be described as a tree itself, which we call a $\gamma$-split-tree. Recall that 
in the hierarchical decomposition we would partition each cluster $C$ into $C_1,\ldots,C_\sigma$ where each $C_i$ has diameter at most $\frac{\Delta_C}{2}$ and $\sigma\in 2^{O(\kappa)}$.
In a $\gamma$-split-tree we consider $\gamma$ many such (independent) partitions of $C$ in each step when we decompose $C$. This will help us to show that for at least one partition, the properties that are shown to hold in expectation (Assumptions 2 and 3) actually hold with high probability for at least one partition of $C$. 
We don't know which of the several parallel decompositions we run at each step will have this property; so we guess and try all of them in our DP.

In the next subsection we formalize this and present our structural theorem that proves the existence of a nearly optimum solution with certain properties. We then show in Subsection \ref{sec:mu-excess} that the structural theorem in fact shows the existence of a $(\eps,\mu)$-approximation for {$k$-stroll}. In Subsection \ref{sec:dp} we show how we can find that using a suitable DP. Finally in Subsection \ref{sec:p2p} we prove that we show an approximation scheme for {P2P orienteering}.

%%%%%%%%%%%%%%%%%%%%%%%%%%%%%%%%%%%%%%%%%%%%%%%%%%%%%%%%%%%%%%%%%%%%%%%%%
\subsection{Structural Theorem}\label{sec:structure}
Let $P$ be any feasible solution and $P^*$ be an optimal solution to {$k$-stroll}. For a subgraph (or subset of vertices) $U\subseteq G$ we use $P_U$ to denote the restriction of $P$ to $U$.

\begin{definition}\label{m}
$P$ is called $\eta$-dense with respect to $U\subseteq G$ if 
$||P_{U}||> \eta \Delta_{U}$; otherwise $P$ is sparse.
\end{definition}

We will set $\eta=\frac{\log n }{\epsilon}$. A set $C\subseteq V$ is
a \emph{$\rho$-cover} for  $U \subseteq  V$ if for any $u \in  U$, there exists $c \in  C$ such that $d(u, c) \leq  \rho$. A set $S$ is a \emph{$\rho$-packing} if for any two distinct points $u$ and $v$ in $S$,
$d(u,v) \geq 2\rho$ holds. A set $N$ is a \emph{$\rho$-net} for  $U \subseteq  V$ if $N$ is a $\frac{\rho}{2}$-packing and a $\rho$-cover for $U$. We can find a $\rho$-net for any graph $G$ easily by iteratively removing balls of radius $\rho$. The center of these balls forms a $\rho$-net.

%by constructing a minimal $\rho$-cover greedily (\cite{talwar2004bypassing}).

Consider a set $C\subseteq G$.
A {\em random partition} of $C$ is a partition  
$\pi_C=C_1,C_2,\ldots,C_\sigma$ of $C$ that is obtained by taking a random $\frac{\Delta_C}{4}$-net of $C$ (i.e. the centers chosen iteratively are picked randomly); so the balls of the net have diameter $\frac{\Delta_C}{2}$ and we know that $|\pi_C|\leq 2^{O(\kappa)}$.
Note that the algorithm for proving Theorem \ref{thm:hier-dec} essentially starts from the single cluster $\{V\}$ (as the root of the decomposition) and at each step, when we have a cluster $C$, it is decomposed into $2^{O(\kappa)}$ parts by taking a random partition of $C$.
%Let $c_i$ be the centre of the ball for $C_i$. 
We say an edge $(u,v)$ is crossing $\pi_{C}$ if $u$ and $v$ are in different parts of $\pi_C$ i.e. $u\in C_i$ and $v\in C_j$ for $i\not=j$. Such an edge is called a {\bf\em bridge-edge}. The following lemma (implied
by Theorem \ref{thm:hier-dec}) shows one property of  bridge-edges 
of $\pi_{C}$.

\begin{lemma}\cite{Talwar04}\label{lem:bridge}
	For any edge $(u,v)$ with $u,v\in C$, the probability that $(u,v)$ crosses $\pi_{C}$ (and so is a bridge-edge) is at most $ \kappa'\frac{d(u,v)}{\Delta_{C}}$ for some constant $\kappa'=O(\kappa)$.
\end{lemma}

Consider $P_C=P\cap C$.
We consider the set of edges of $P$ that are bridge-edges, denoted as $E_{\pi_{C} } $, i.e. $ E_{\pi_{C} }=\{(u,v)\in P: (u,v) \text{ crosses }   \pi_{C}  \} $. Suppose $P_C$ is sparse with respect to $C$, i.e. $||P_C||\leq\eta\Delta_C$. 
We upper bound the number of bridge-edges of $\pi_C$, i.e. the size of $E_{\pi_{C} } $ in this case. 

\begin{lemma}\label{lem:sparse}
	If $P$ is sparse with respect to $C$, then $\E[|E_{\pi_{C} } |]\leq \frac{\kappa'\log n}{\epsilon} $ for some constant $\kappa' =O(\kappa)$.
\end{lemma}
\begin{proof}
	Recall that the probability $(u,v)$ crosses $\pi_{C} $ is at most $ \kappa' \frac{d(u,v)}{\Delta_{C}}$ for some constant $\kappa'=O(\kappa)$. Therefore, $ \E (|E_{\pi_{C} } | ) =\sum_{u,v\in P_C} \kappa' \frac{d(u,v)}{\Delta_{C}}  =  \kappa'\frac{||P_{C}||}{\Delta_{C}} $. Since $P$ is sparse with respect to $C$, thus $ \kappa'\frac{||P_{C}||}{\Delta_{C}}\leq \kappa' \frac{\eta\Delta_{C}}{\Delta_{C}} =\kappa'\eta=\kappa'\frac{\log n}{\epsilon}  $.
\end{proof}

Now suppose $P$ is $\eta$-dense with respect to $C$. Again consider
the (random) partition $\pi_C$ of $C$. For each $C_i\in \pi_C$, we further consider a $\beta\Delta_{C_i}$-net of it where $ \beta=\frac{\epsilon}{4\kappa' \delta}$ ( recall $\delta=\log \Delta_{G}$ and $\kappa'$ is the constant in Lemma \ref{lem:bridge}), denoted as $S_i$ and we call them the portal set for $C_i$. The following packing property of doubling metrics bounds
the size of $|S_i|$:

\begin{proposition}[\cite{gupta2003bounded}]\label{prop:packing}
Let $(V, d)$ be a doubling metric with doubling dimension $\kappa$ and
diameter $\Delta$, and let $S$ be a $\rho$-packing. Then $|S| \leq \left(\frac{\Delta}{\rho}\right)^{\kappa}$.
\end{proposition}

Using this proposition: $|S_i|\leq (\frac{8\kappa' \delta}{\epsilon})^{\kappa}$. We are going to modify $P_C$ such that when it crosses between different clusters $C_i$ in $\pi_C$ (i.e. uses bridge-edges) it does so via portals only. This means after the modification it will use {\bf\em portal-edges}; those bridge-edges whose both end-points are portals.
We say a path $P'$ is portal respecting with respect to $\pi_{C}$ if for any edge $u,v$ of $P'$ crossing between two parts $C_i,C_j\in \pi_{C} $, it only cross through portals, i.e. $u\in  S_i$ and $v\in  S_j$. The following lemma shows one can modify $P_C$  to be portal respecting with respect to $\pi_{C}$ with a small increase of length.

\begin{lemma}\label{lem:dense}
	$P_C$ can be changed to another path $P'$ that is portal respecting with respect to $\pi_{C}$ such that $\E[||P'||]\leq (1+\frac{\eps}{2\delta})||P_C||$.
\end{lemma}
\begin{proof}
	We start with $P_C$ and whenever an edge $u,v$ of $P_C$ is crossing two parts $C_i,C_j\in\pi_C$ that are not portals we replace that edge with a path between $u,v$ via the closest portals in $C_i$ and $C_j$.
	Consider any edge $u,v$ in $P_C$ that crosses $\pi_{C}$, say $u\in C_i$ and $v\in C_j$. Let $u'$ be the nearest portal to $u$ in $S_i$ and $v'$ be the nearest portal to $v$ in $S_j$. Replace edge $u,v$ in $P_C$ by the edges $(u,u')$, $(u',v')$, $(v',v)$. The increased length incurred is $d(u,u')+d(v,v')+d(u',v')-d(u,v)$, which is at most $ 2d(u,u')+2d(v,v')$ by triangle inequality. Note that because $S_i$ is a $\beta \Delta_{C_i} $-net of $C_i$ then  $d(u,u')\leq \beta \Delta_{C_i} \leq \beta \frac{\Delta_{C}}{2} $ and  $d(v,v')\leq \beta \Delta_{C_j}\leq \beta \frac{\Delta_{C}}{2} $. Thus the increased length incurred is at most 
	$2\beta \Delta_{C}$.
	Recall that for $u,v\in P_C$, the probability that $(u,v)$ crosses 
 $\pi_{C}$ is at most $\delta\kappa\frac{d(u,v)}{\Delta_{C}}$. Thus in expectation,  the increased length of $P_C$ after making it portal respecting with respect to $\pi_{C}$ is at most $\sum_{u,v\in P_{C}} \kappa'\frac{d(u,v)}{\Delta_{C}}2\beta \Delta_{C} =\kappa'\beta ||P_{C}||= \frac{\epsilon}{2\delta} ||P_{C}|| $.
\end{proof}

Now we formalize the idea of parallel runs of the hierarchical decomposition of Theorem \ref{thm:hier-dec}. Note that the algorithm for proving Theorem \ref{thm:hier-dec} starts from the single cluster $\{V\}$ (as the root of the decomposition) and at each step uses a {\em random partition} to the current cluster $C$ to decompose it   into $2^{O(\kappa)}$ clusters of diameter half the size.
This continues until we arrive at singleton node clusters (leaves of the split-tree).
We call one random partition of a cluster $C$ a {\em split} operation.
The $\gamma$-split-tree hierarchical decomposition of a doubling metric is obtained by 
considering $\gamma$ random partitions of $C$ (instead of just one) at each step. 
A $\gamma$-split-tree decomposition has two types of nodes that appear in alternating layers: cluster nodes and split nodes.
For each cluster $C$ we have $\gamma$ {\em split} nodes in the $\gamma$-split-tree 
(each corresponding to a random partition of $C$) and these split nodes are all children of $C$. For each split node $s$ that is a child of $C$, it will have children $C_1,\ldots,C_\sigma$ that are clusters obtained by decomposing $C$ according to the random partition corresponding to $s$. We continue until leaf nodes are cluster nodes with constant size $a=O(1)$.
%(for some constant $a$ to be specified)

\begin{definition}
	A $\gamma$-split tree for $G$ is a rooted tree $\Gamma$ with alternating
	levels of cluster nodes and split nodes.
	A cluster node $C$ corresponds a subset of $V$. The root of $\Gamma$, corresponds to the single cluster $\{V\}$ and each leaf node corresponds to a set of size $a=O(1)$ of vertices.
		Each non-leaf cluster node $C$ has $\gamma$ split nodes as its children. 
		Each split node corresponds to a (random) partition of $C$ into clusters
		of diameter $\frac{\Delta_C}{2}$. Each of those clusters become cluster children
		of the split node. So each split node has $2^{O(\kappa)}$ (cluster) children.
\end{definition}

Suppose $\Gamma$ is a $\gamma$-split-tree for $G$ and $\Phi$ is a partial function from (some) cluster nodes of $\Gamma$ to one of their split node children with the following properties:
\begin{itemize}
    \item $\Phi(C_0)=s$ for some child split node of $C_0$ (where $C_0$ is the root of $\Gamma$).
    \item if every ancestor split node of $C$ is in the image of $\Phi$ 
    then $\Phi(c)$ is defined.
\end{itemize}

Note that this function induces a "standard" hierarchical decomposition split-tree as in Theorem \ref{thm:hier-dec} as follows: 
start at the root node $C_0=\{V\}$ of $\Gamma$ and let root of tree $T$ be $C_0$ and repeat the following procedure: at each step, being at a cluster node $C$ pick $\Phi(C)$ (a split node child of $C$) and consider the cluster children of $\Phi(C)$, say $C_1,\ldots,C_\sigma$, and create nodes corresponding to these as children of $C$ in $T$. Recursively repeat the procedure from each of them. This builds a split-tree tree $T$. 
We call such a (partial) function $\Phi$ a {\em determining} function for $\Gamma$ and we say $T$ is induced by $\Phi$ on $\Gamma$: $T=\Gamma|_\Phi$.

%Given a $\gamma$-split tree $\Gamma$ and a determining function $\Phi$, and the induced hierarchical decomposition $\gamma|_{\Phi}$, a window is a subgraph of $G$ defined by cluster node $C$ in $\Gamma|_{\Phi}$.
%    We denote a window (subgraph) defined by a cluster $C$ as $w_C$.
For any cluster $C$ and the split node defined by $\Phi(C)$, we use 
$|P^*\cap E_{\pi_{\Phi(C)}}| $ to denote the number of bridge-edges of $P^*$  in $\pi_{\Phi(C)}$, i.e. $| P^*\cap R_{\pi_{\Phi(C)}}|=E_{ \pi_{\Phi(C)} } $.

Let $\gamma=3\log n$, we can build a $\gamma$-split tree $\Gamma$ for $G$ in the following way.	We start by making $C_{0}=\{V\}$ to be the root of $\Gamma$ and iteratively add levels to $\Gamma$. For a cluster node $C$, we generate its children split nodes as follows:	
we compute $\gamma$ independent random partitions of $C$, denoted as 
$\{\pi_{i}\}$, $1\leq i\leq\gamma_i$. Each $\pi_i=C_{i,1},\ldots,C_{i,\sigma_i}$ (where $\sigma_i\leq 2^{O(\kappa)}$) is obtained by taking a
random $\frac{\Delta_{C}}{4}$-net for $C$. Let $R_{\pi_{i}} $ be the set of bridge-edges of $C$, i.e. edges in $C$ crossing $\pi_{i}$: $R_{\pi_{i}} =\{u,v\in C: u\in C_{i,j}, v\in C_{i,j'}, \text{ for } j\not=j'\}$. For each $C_{i,j}$, we further consider the portal set $S_{i,j}$ which is a $\beta\Delta_{C_{i,j}}$-net of it. Let $R'_{\pi_{i}}$ be the set of portal edges of $C$: $R'_{\pi_{i}}=\{ u,v\in C: u\in S_{i,j}, v\in S_{i,j'}, \text{ for } j\not=j' \}$. Since $G$ is a doubling metric and $S$ is a net, we find the following bound on the sizes of the portals: 
$|R'_{\pi_{i}}|\leq(\sigma_i|S_{i,j}|)^{2}$ which is bounded by
$(\frac{16\kappa'\delta}{\epsilon})^{2\kappa} $. We create a child split node $s_{i}$ for $C$ for each $\pi_{i}$. 

For a split node $s$, let $C$ be its parent cluster node and let  $\pi_i$ be the partition (i.e. the $\frac{\Delta_{C}}{4}$-net for $C$) corresponding to $s$. Then for each $C_{i,j}\in \pi_i$ we create a child cluster node $C_j$ for $s$. The number of the children cluster nodes of $s$ is at most $2^{O(\kappa)}$.	
We continue this process until each leaf node is a cluster node $C$ with $|C|\leq a$ for some constant $a$.
From the construction above we know $\Delta_{C_i}\leq \frac{\Delta_{C}}{2}$
if $C_i$ is a child of a split node $s$ which is a child of $C$.
Thus there are at most $2 \log \Delta_{G}= 2\delta$ levels in $\Gamma$. If we define the height of $\Gamma$ as the number of levels of cluster nodes then the height of $\Gamma$ is at most $\delta$. The branching factor of $\Gamma$ is then the product of the branching factor of a cluster node and a split node, which is at most $\gamma 2^{O(\kappa)}$. Hence the size of $\Gamma$ is $(3\log n 2^{O(\kappa)})^{\delta}$.
Now we are ready to prove the following structure theorem for a near optimum solution.

\begin{theorem}(structure theorem)\label{thm:structure}
    Let graph $G=(V,E)$ with doubling dimension $\kappa$, start and end nodes $s,t\in V$, and integer $k$ be given as an instance of $k$-stroll. Assuming $P^*$ is an optimum solution and $\gamma=3\log n$ we can construct a $\gamma$-split-tree where with probability at least $1-\frac{1}{n}$ there exists a determining function 
    $\Phi$ and a corresponding split-tree hierarchical decomposition 
    $T=\Gamma|_\Phi$ of $G$ and a nearly optimal solution $P'$ such that $P'$ visits at least $k$ vertices and for any cluster $C\in T$ we have either:
    \begin{itemize}
        \item 1) $ |P'\cap R_{\pi_{\Phi(C)}} |\leq 2\kappa'\frac{\log n }{\epsilon}$, or        
    \item  2) $|P'\cap R'_{\pi_{\Phi(C)}} |\leq  ( \frac{16\kappa'\delta}{\epsilon} )^{2\kappa}$ and
        $||P'\cap C ||\leq(1+\epsilon) ||P^*\cap C||$.
    \end{itemize}
\end{theorem}

\begin{proof}	
Suppose $\Gamma$ is a $\gamma$-split-tree for $G$.
We build $P'$ iteratively based on $P^*$ and at the same time build $\Phi(\cdot)$ and hence $T$ from the top to bottom.

Initially we set $P'$ to be $P^*$ and start from the root cluster node $C_0$ of $\Gamma$. At any point when we are at a cluster node $C$
we consider whether $P'$ is sparse or dense with respect to $C$:\\
\textbf{Sparse case:} If $||P'\cap C||\leq\eta\cdot \Delta_{C}$, i.e. $P'$ is $\eta$-sparse with respect to $C$, we don't modify $P'\cap C$ 
for when going down from $C$ to any split node $s$ of $C$.
Consider any child split node $s$ of $C$ and let $\pi_s$ be the partition of
$C$ according to the split node $s$.
Consider the number of edges of $P'\cap C$ that are bridge-edges based on this partition i.e. $|E_{\pi_s}|$ (where $E_{\pi_s}=P'\cap R_{\pi_s}$). According to Lemma \ref{lem:sparse},  
$\E(|P'\cap R_{\pi_{s}}|)=\E(|E_{\pi_{s}}|)\leq\kappa'\frac{\log n}{\epsilon}$. Let the event $\Lambda_s$ be the event that $|P'\cap R_{\pi_{s}}| \leq 2 \kappa'\frac{\log n }{\epsilon}$.
By Markov inequality $\Pr[\Lambda_s]\geq \frac{1}{2}$. 
Recall there are $\gamma=3\log n$ many children split nodes  of $c$, i.e. $\gamma$ independent random partitions of $C$. Thus the probability that for at least one child split node $s$ of $c$, event $\Lambda_s$ holds is at least
$1- (1-\frac{1}{2})^{3\log n}\geq 1-\frac{1}{n^{3}} $. 
In this case we select split node $s$ and define $\Phi(C)=s$ and consider each cluster child node of $s$ iteratively ($P'$ has not changed
in this step).\\
\textbf{Dense case:} Suppose $P'$ is $\eta$-dense with respect to $C$.
Consider an arbitrary split child $s$ of $C$. We will modify $P'$ 
going down the split node $s$ to be portal respecting with respect to 
$\pi_{s}$ as described in Lemma \ref{lem:dense}. Assume $\pi_s=C_1,\ldots,C_\sigma$. For each $C_i$ let $S_i$ be a set of portals.
If we go down the split node $s$ we modify $P'$ to be portal-respecting
for each $C_i$. Note that the set of edges $P'$ crossing $\pi_{s}$ after making it portal respecting with respect to $\pi_{s}$ is a subset of $R'_{\pi_{s}}$. Thus $ |P'\cap R'_{\pi_{s}}| $ is at most $|R'_{\pi_{s}}|$ which is at most  $(\frac{16\kappa' \delta}{\epsilon} )^{2\kappa} $ in this case.  According to Lemma \ref{lem:dense}, the expected increase of length of $P'$ after making it portal respecting with respect to $\pi_{s}$ is at most $\frac{\epsilon}{2\delta} ||P'\cap C||$. Let $\Lambda'_{s}$ be the event that the increase of length of $P'$ after making it portal respecting with respect to $\pi_{s}$ is at most $\frac{\epsilon}{\delta}||P'\cap C||$.
By Markov inequality,  $\Pr[\Lambda'_s] \geq \frac{1}{2} $.  Recall there are $\gamma$ many children split nodes of $C$. Thus the probability that for at least one child split node  $s$ of $C$  we have $\Lambda'_{s}$ is at least $1- (1-\frac{1}{2})^{3\log n}\geq 1-\frac{1}{n^{3}} $. 
In this case we set $\Phi(C)=s$ for this particular split node $s$ for which
$\Lambda'_s$ happens.

Therefore, regardless of whether $P'$ is sparse or dense w.r.t. $C$,
such $s$ exists for cluster $C$ with the probability at least 
$1-\frac{2}{n^{3}}$ and we can define $\Phi(C)$. Once we have
$\Phi(.)$ defined for clusters at a level of $\Gamma$, 
we have determined the clusters at the same level of $T=\Gamma|_\Phi$.
Note there are at most $n$ cluster nodes in one level of $T$. Thus with probability at least 
 $ 1-\frac{2}{n^{2}} $  such split nodes exist for all cluster nodes in one level. 
Since height of $\Gamma$ (and $T$) is $\delta$, thus with probability at least $(1-\frac{2}{n^{2}} )^{\delta}\geq 1-\frac{1}{n} $ (assuming that $\delta$ is polylogarithmic in $n$)
such $\Phi(.)$ is well defined over all levels. 

Note the increase of length of $P'$ only occurs when $P'$ is $\eta$-dense with respect to $C$ and for any of these clusters $C$ in 
$\Gamma|_{\Phi}$, the increase of length $P'$ by modifying $P'$ to be portal respecting with respect to $\pi_{\Phi(C)}$ is at most $ \frac{\eps}{\delta} ||P'\cap C||$. Since this $(1+\frac{\eps}{\delta})$-factor
increase occurs each time we go down the decomposition form a cluster $C$ to the next cluster level down, and 
the height of the decomposition is $\delta$,
thus inductively, for any cluster $C$ in 
$\Gamma|_\Phi$: $||P' \cap C||\leq (1+ \frac{\epsilon}{\delta} )^{\delta}||P^*\cap C||\leq  e^{\epsilon} || P^*\cap C|| \leq (1+\epsilon')   || P^*\cap C|| $ for some $\eps'>0$ depending on $\eps$. 
Replacing $\epsilon$ with $\eps'$ we get 
$||P'\cap C||\leq (1+\epsilon)   || P^*\cap C|| $.
\end{proof}

%%%%%%%%%%%%%%%%%%%%%%%%%%%%%%%%%%%%%%%%%%%%%%%%%%%%%%%%%%%%%%
\subsection{$(\eps,\mu)$-approximation for {$k$-stroll}}\label{sec:mu-excess}
In this section we show a stronger bound on the length of the near optimum solution $P'$ guaranteed by Theorem \ref{thm:structure}.
Recall that a $(\eps,\mu)$-approximation for {$k$-stroll} is a path $P$ with $|P|=k$ where $||P||\leq ||P^*||+\eps\cdot\calE_{P^*,\mu}$ where $P^*$ is the optimum solution. We prove that path $P'$ in Theorem \ref{thm:structure} is in fact an $(\eps,\mu)$-approximation for $\mu=\lceil \frac{1}{\epsilon} \rceil+1$ .
%More specifically, we show $P'$ is a $(\epsilon,u)$-approximation for the {$k$-stroll} instance in $G$ for $u=\lceil \frac{1}{\epsilon} \rceil+1$ . 

Recall that in the proof of the structure theorem, the increase in length of $P'$ only happens in cases when $P'$ is dense with respect to a cluster $C$ in $\Gamma$ and we make the path portal respecting. Consider such a dense cluster $C$ in the hierarchical decomposition $T=\Gamma|_{\Phi}$, i.e. $P'$ is 
$\eta$-dense with respect to $C$. We show $P'$ has high $\mu$-excess in this case and the increased length of $P'$ in $C$ can be upper bounded by a factor of $\mu$-excess of $P^*$. 

\begin{lemma}\label{lem:muexcess}
    Let $D$ be a set of disjoint clusters and $Q$ be a path. Then $\calE_{Q,2}\geq \sum_{C\in D}(||Q \cap C||-\Delta_{C}) $.
\end{lemma}
\begin{proof}
    Intuitively, this is saying if $Q$ passes through several (disjoint) clusters in $D$ then the excess of $Q$ is at least as big as sum of excess of $Q$ in those clusters.

    Suppose start-end node of $Q$ are $u,v$ and let 
    $Q_{0}$ be the path just consisting of $u,v$. By definition of the excess, $\calE_{Q,2}=||Q||-||Q_{0}||$. Now consider following path $Q'$, which starts at $u$ and follows $Q$ but when it encounters a cluster $C$ in $D$ and it visits $C$ for the first time it directly connects the start and end node of the subpath of $Q$ in $C$. When it encounters a cluster $C$ in $D$ that is visited before, then bypasses $C$ entirely, i.e. directly connects the last vertex in $Q$ before it enters $C$ this time and the first vertex in $Q$ after it visits after it exits $C$ this time. From the construction of $Q'$, if $C\in D$, then $||Q'\cap C||\leq \Delta_{C}$. Clearly for any cluster $C \notin D$: $||Q'\cap C ||=||Q\cap C||$.  Thus $\calE_{Q,2}=||Q||-||Q_{0}|| \geq ||Q||-||Q'||= 
    \sum_{C\in D}(||Q \cap C||- ||Q' \cap C||) \geq \sum_{C\in D}(||Q \cap C||- \Delta_{C}) $.
\end{proof}

\begin{theorem}\label{thm:muexcess}
 Let graph $G=(V,E)$ with doubling dimension $\kappa$, start and end nodes $s,t\in V$, and integer $k$ be given as an instance of $k$-stroll. Suppose $P'$ and $\Gamma|_{\Phi}$ are as guaranteed by Theorem \ref{thm:structure}. Let $\mu=\lceil \frac{1}{\epsilon}\rceil+1$, then $P'$ is a $(\epsilon,\mu)$-approximation.
\end{theorem}

\begin{proof}
Note that in the proof of Theorem \ref{thm:structure}, the increase of length of $P'$ only occurs when $P'$ is $\eta$-dense with respect to a cluster $C$. We generate a set of disjoint clusters $D$ in $T=\Gamma|_{\Phi} $: we start from $C_{0}$ (the root of $T$) to generate $D$ iteratively. If $P'$ is $\eta$-dense with respect to $C$ then add $C$ to $D$, if $P'$ is sparse with respect to $C$ and $C$ is a non-leaf cluster node then iteratively consider all children cluster nodes of $C$.
Let $D$ be the set of cluster nodes returned by this process. From the construction of $D$, clusters in $D$ are disjoint, and for each cluster $C\in D$, $P'$ is $\eta$-dense with respect to $C$ and $ ||P'||-||P^*||= \sum_{C\in D} (||P' \cap C ||-||P^* \cap C || )  $.

Let  $\{ v_{1},\cdots,v_{\mu} \}$  be the optimal $\mu$-jump of $P^*$  and let  $P^*_{1},P^*_{2},\cdots,P^*_{\mu-1}$ be the subpaths divided by the vertices in this $\mu$-jump, i.e. $P^*_{i}$ is the subpath of $P^*$ whose start and end are $v_{i}$ and  $v_{i+1}$, respectively. By definition of excess, $  \calE_{P^*,\mu}=\sum_{i=1}^{\mu-1}\calE_{P^*_{i},2}$. For the set $D$ and each $P^*_{i}$, by Lemma \ref{lem:muexcess}, $ \calE_{P^*_{i},2} \geq \sum_{C\in D} (||P^*_{i}\cap C||-\Delta_{C}) $. Thus $\calE_{P^*,\mu}\geq  \sum_{i=1}^{\mu-1}\sum_{C\in D}(||P^*_{i}\cap C||-\Delta_{C})$, which is $\sum_{C\in D}(||P^*\cap C||-(\mu-1)\Delta_{C}) $. Recall from the proof of Theorem \ref{thm:structure} that we modify $P'$ when
it is $\eta$-dense with respect to $C$ (i.e. $||P'\cap C||>\frac{\log n}{\eps}\Delta_C$) and we always have 
$||P'\cap C||\leq (1+\eps)||P^*\cap C||$, which implies
$||P^*\cap C||\geq \frac{\log n}{\epsilon(1+\eps)} \Delta_{C} $; in a sense $P^*$ is also (almost) $\eta$-dense with respect to $C$. Since 
$\mu=\lceil \frac{1}{\epsilon} \rceil +1 $, thus $||P^*\cap C||-(\mu-1)\Delta_{C} \geq \frac{||P^*\cap C||}{2\eps}  $ and $ \calE_{P^*,\mu}\geq \sum_{C\in D}\frac{||P^*\cap C||}{2\eps}$. 

As in the proof Theorem \ref{thm:structure}, for any cluster 
$C\in T$, $||P'\cap C||\leq (1+\epsilon)||P^*\cap C||$. Thus  
$ ||P'||-||P^*||= \sum_{C\in D} (||P' \cap C ||-||P^* \cap C || )  \leq  \sum_{C\in D} \epsilon||P^* \cap C || \leq 2\epsilon^2\calE_{P^*,\mu}\leq\eps\calE_{P^*,\mu} $. \end{proof}

We should note that we can generalize this proof slightly as follows. Suppose that we pick some arbitrary $\mu$-jump of $P^*$ (instead of the optimum) and consider $\tilde{P}^*_1,\ldots,\tilde{P}^*_{\mu-1}$ which are the subpaths of $P^*$ defined by that $\mu$-jump. Then the same arguments show that $\sum_{i=1}^{\mu-1}\calE_{\tilde{P}^*_i,2}\geq \sum_{i=1}^{\mu-1}\sum_{C\in D}(||\tilde{P}^*_i\cap C||-\Delta_C)=\sum_{C\in D}||P^*\cap C||-(\mu-1)\Delta_C$ which implies
$\sum_{i=1}^{\mu-1}\calE_{\tilde{P}^*_i,2}\geq \sum_{C\in D}\frac{||P^*\cap C||}{2}$. Using this
one can show that at the end $||P'||\leq ||P^*|| +\eps\sum_{i=1}^{\mu-1}\calE_{\tilde{P}^*_i,2}$.
This slightly more general version will be used later when designing our algorithm for deadline TSP.

%%%%%%%%%%%%%%%%%%%%%%%%%%%%%%%%%%%%%%%%%%%%%%%%%%%%%%%%%%%%%%%%
\subsection{Finding a near optimum solution for {$k$-stroll}}\label{sec:dp}
In this section we prove Theorem \ref{thm:kstrollDBL} by showing
how we can find a near optimum solution for {$k$-stroll} as guaranteed by Theorem \ref{thm:muexcess} using Dynamic Programming (DP).
The DP is built on the $\gamma$-split tree $\Gamma$ we compute.
Consider an arbitrary cluster node $C$ in $\Gamma$ and the restriction of the near optimum solution $P'$ in the subgraph $G(C)$, denoted by $P'_C$. 
The set of edges of $P'_C$ might be a collection of disjoint paths; one can imagine following along $P'$, it enters and exits $C$ multiple times and each time it enters it follows a path in $C$ until it exits again (assuming that $s,t$ are not in $C$). If we denote the start-end of these subpaths in $C$ as $s_i,t_i$, $1\leq i\leq m$ (for some $m$) and if $|P'_C|=k_C$ then $P'_C$ is a feasible solution for {multi-path $k$-stroll} with start-end pairs $s_i,t_i$ and parameter $k_C$.
This suggests a subproblem in our DP table will corresponds to an instance of {multi-path $k$-stroll} in a cluster $C$.

We define a subproblem in the DP as an instance of {multi-path $k$-stroll} with specified cluster node $C$, integer $k_{C}$ and $\sigma_{C}$  start-end pairs of vertices $\{(s_{i},t_{i})\}$ and the goal is to find a set of paths 
$\{P_{i} \}_{i=1}^{\sigma_{C}}$ such that $P_{i}$ is a $s_{i}$-$t_{i}$ path in $C$ and $ |\cup_{i=1}^{\sigma_{C}} P_{i}| =k_{C} $ while minimizing  $\sum_{i=1}^{\sigma_{C}}||P_{i}|| $. We use $A[C,k_C,(s_{i},t_{i})_{i=1}^{\sigma_{C}}] $ to denote the subproblem defined above and let the entry of the table store the optimal value of the solution to {multi-path $k$-stroll} for this subproblem.
 One should note that in the proof of Theorem \ref{thm:structure}, when we build $P'$ from $P^*$, each time we go down a node of the split-tree, we might have a number of bridge-edges (when $P'$ is sparse with respect to the current cluster $C$) or portal-edges (when $P'$ is dense with respect to the current cluster $C$). When $P'$ is sprase,
the number of bridge-edges is at most $2 \kappa'\frac{\log n}{\eps}$ (it was event $\Lambda_s$ in the sparse case) and when $P'$ is dense then the number of portal-edges is at most $(\frac{16\kappa' \delta}{\epsilon} )^{2\kappa}$. Therefore, each step $P'$ restricted to $C$ might be chopped up into at most $(\frac{16\kappa' \delta}{\epsilon} )^{2\kappa}$ subpaths, each with a new start-end point in a cluster $C_i$ that is part of the partition of $C$. Given that the height of the split-tree is $\delta$, the total number
of start-end points for the instance of {multi-path $k$-stroll} at any cluster node $C$ is at most $\delta(\frac{16\kappa' \delta}{\epsilon} )^{2\kappa}$. This upper bounds $\sigma$ in our subproblems.

Now we describe how to fill in the entries of the table.
The base cases are when the cluster $C$ has constant size $|C|=a$. Such instances can be solved using exhaustive search in $O(1)$ time. 
 
Consider an arbitrary entry $A[C,k_C,(s_i,t_i)_{i=1}^{\sigma_C}]$ where for all split nodes children of $C$ and every cluster children of them the entries of the table are computed.
Consider any child split node $s$ of $C$ in $\Gamma$
and let $C_{1},C_{2},\cdots C_{g}$ be the children cluster nodes of $s$ in 
$\Gamma$. Recall $\pi_{s}$  is the corresponding partition of $C$ and $R_{\pi_{s}}$ is the set of bridge-edges in $C$ (crossing $\pi_{s}$). For each $C_{j}$ let $S_{j}$ be its portal set and recall $R'_{\pi_{s}}$ is the set of edges crossing $\pi_{s}$ only through $\{S_{j}\}$. We guess $k_{C_{j}}$ for each $C_{j}$ such that $\sum_{j=1}^{g}k_{C_{j}}=k_{C}$. We show how to guess start-end pairs $\{ (s_{i},t_{i})_{i=1}^{\sigma_{C_{j}}}  \}$ for each $C_{j}$ and check the consistency of them. To do so we consider two cases: in the first case (meaning we guess we are in the sparse case) we guess a subset of $R_{\pi_{s}}$ of size at most $2\kappa'\frac{\log n}{\epsilon} $; in the second case (meaning we assume we are in the dense case) we guess a subset of $R'_{\pi_{s}}$ of size at most
$(\frac{16\kappa' \delta}{\epsilon} )^{2\kappa}$.
Let $E_{\pi_{s}} $ be the subset guessed in either case. Furthermore for each edge in $E_{\pi_{s}}$, we guess it is in which one of the $\sigma_{C}$ paths with start-end pair $(s_{i},t_{i})_{i=1}^{\sigma_{C}}$ and for each path with start-end pair $(s_{i},t_{i})$ we guess the order of the guessed edges appearing on the path. Specifically speaking, let $e_{1},e_{2}\cdots,e_{l}$  be the edges guessed in the path with source-sink pair $ (s_{i},t_{i})$ appearing in this order. Let $C_{a_1},C_{a_2},\cdots,C_{a_{l+1}} $ be the children cluster nodes of $s$ that the path encounters following $e_{1},e_{2}\cdots,e_{l} $, i.e. $e_{1}$ crosses between $C_{a_1}$ and $C_{a_2}$,  $e_{2}$ crosses between $C_{a_2}$ and $C_{a_3}$, $\cdots$, and $e_{l}$ crosses between $C_{a_l}$ and $C_{a_{l+1}}$. Then we set  $s_{i}$ and the endpoint of $e_{1}$ in $C_{1}$ to be a start-end pair in $C_{a_1}$, the endpoint of $e_{1}$ in $C_{a_2}$  and the  endpoint of $e_{2}$ in $C_{a_2}$ to be a start-end pair in $C_{a_2}$, $\cdots$, the endpoint of $e_{l}$ in $C_{a_{l+1}}$ and $t_{i}$ to be a start-end pair in $C_{a_{l+1}}$. By doing so we generate start-end pairs for each $C_{j}$ and we sort them based on their ordering in 
$s_{i}$-$t_{i}$ path and in the increasing order of $i$. This defines 
$\sigma_{C_{j}}$ start-end pairs for each $C_{j}$.

\begin{lemma}\label{lem:kstroll-dp}
    We can compute all entries  $A[C,k_C,(s_{i},t_{i})_{i=1}^{\sigma_{C}}]$
    in time $ n^{ O( (\frac{\delta}{\epsilon})^{2\kappa+1} )} $.
\end{lemma}

\begin{proof}
Formally, to compute $A[C,k_C,(s_{i},t_{i})_{i=1}^{\sigma_{C}}]$:

\begin{itemize}
\item  Consider any child split node $s$ of $C$, let $C_{1},C_{2},\cdots,C_{g}$ be the children cluster nodes of $s$ (where $g$ can depend on $s$).
\item  Guess (i.e. try all possible values) $k_{C_j}$ for each $C_{j}$ such that $\sum_{j=1}^{g}k_{C_{j}} =k_{C}$. 
\item  Let $\pi_{s}$ be the corresponding partition of $C$ and 
$R_{\pi_{s}}$ be the bridge-edges of $\pi_{s}$. For each $C_{j}$ let $S_{j}$ be the portal set for it and let $R'_{\pi_{s}}$ be the set of portal-edges of $\pi_{s}$. We consider both of the following two cases: 1) we guess a subset of $R_{\pi_{s}}$ of size at most $2\kappa'\frac{\log n}{\epsilon}$; 2) we guess a subset of $R'_{\pi_{s}}$; in both cases we denote the set of guessed edges as $E_{\pi_{s}}$.

\item For each edge in $E_{\pi_{s}}$, we guess it is in which one of the $\sigma_{C}$ paths with start-end pair $(s_{i},t_{i})_{i=1}^{\sigma_{C}}$ and for each path with start-end pair $(s_{i},t_{i})$ we guess the order of the guessed edges appearing as described above. 
We generate $\{ (s_{i},t_{i})\}_{i=1}^{\sigma_{C_{j}}}$ for each $C_{j}$ accordingly. Then:

\item We set $A[C,k_C,(s_{i},t_{i})_{i=1}^{\sigma_{C}}]=\min
\sum_{j=1}^{g} A[C_{j},k_{C_j},(s_{i},t_{i})_{i=1}^{\sigma_{C_{j}}}] +
    \sum_{(u,v)\in E_{\pi_{s}}}d(u,v) $,
where the minimum is taken over all tuples 
$s, k_{C_{1}},\cdots,k_{C_{g}},(s_{i},t_{i})_{i=1}^{\sigma_{C_{1}}},\cdots,(s_{i},t_{i})_{i=1}^{\sigma_{C_{g}}}$ as described above.

\end{itemize}

Based on the structure theorem and the recurrence given above it is straightforward to see that this recurrence computes the best $P'$ as guaranteed in Theorem \ref{thm:muexcess}. 

Now we analyze the running time. 
%Note $s_{i},t_{i}\in V(c)$  thus
%$\sigma_{C}$ is at most $a^{2}$ in this case. We can enumerate all possible collections of $\{P_{i}\}_{i=1}^{\sigma_{c}}$ such that $P_{i}$ is a $s_{i}$-$t_{i}$ paths. Specifically speaking, we guess a subset of $V(c)$ denoted as $U$, which are at most $2^{a}$ many. Then for each vertex in $U$ we guess it is in which one of the $\sigma_{c}$ path with source-sink pair $(s_{i},t_{i})_{i=1}^{\sigma_{c}}$. For each path with source-sink pair $(s_{i},t_{i})$ we guess the order of vertices appearing on the path. There are at most $|U|!|U|^{\sigma_{c}}$ guessings which are at most $a!a^{a^{2}}$. Among these enumeration of $ \{P_{i}\}_{i=1}^{\sigma_{c}} $, which is at most $2^{a}a!a^{a^{2}} $ many,  we consider the one such that  $|P_{1}\cup\cdots\cup P_{\sigma_{c}}|=k_{c}$  with minimized  $\sum_{i=1}^{\sigma_{c}}||P_{i}|| $.
First we show the time required to compute one entry of the DP table is  %most $(\kappa h+2)\log n \cdot n^{2^{\kappa} }n^{(\frac{4\delta\kappa h}{\epsilon})^{2\kappa}(\log n)^{2\kappa-1}  }   (\frac{4\delta\kappa h\log n}{\epsilon} )^{2\kappa} ! (\frac{4\delta\kappa h \log n}{\epsilon} )^{2\kappa h\log n ( \frac{4\delta\kappa h \log n}{\epsilon} )^{2\kappa}  }    $, which 
$n^{O(  (\frac{\delta}{\epsilon})^{2\kappa+1} )  }$.
In the recursion, for a cluster node $C$, there are $3\log n $ children split nodes of $C$ to consider. For a certain split node $s$, let $C_{1},C_{2},\cdots,C_{g}$ be children cluster nodes of $s$, there are at most $n^{g}$ guesses to for $\{k_{C_{j}}\}$ such that $\sum_{j=1}^{g}k_{C_{j}}=k_{C} $, which is at most $n^{2^{O(\kappa)}}$ because $g\leq 2^{O(\kappa)}$. For $E_{\pi_{s}}$: there are two cases, if $E_{\pi_{s}}\subset R_{\pi_{s}}$ then $|E_{\pi_{s}}|\leq 2 \kappa' \frac{\log n}{\epsilon} $, there are at most 
$n^{(4\kappa'\frac{\log n}{\epsilon})}$ many possible $E_{\pi_{s}}$'s to consider; if $ E_{\pi_{s}}\subset R'_{\pi_{s}}  $, then because $|R'_{\pi_{s}}|\leq ( \frac{16\kappa' \delta}{\epsilon} )^{2\kappa}    $ in this case there are at most $2^{( \frac{16\kappa' \delta}{\epsilon} )^{2\kappa}    }\leq n^{(\frac{16\kappa' \delta}{\epsilon})^{2\kappa}}$ many possible $E_{\pi_{s}}$'s to consider. To generate $(s_{i},t_{i})_{i=1}^{\sigma_{C_{j}}}$ for each $C_{j}$: for a certain $E_{\pi_{s}}$ and for each edge in $E_{\pi_{s}}$ we guess it is in which one of $\sigma_{C}$ paths with start-end pair $(s_{i},t_{i})_{i=1}^{\sigma_{C}}$ and for each path with start-end pair $(s_{i},t_{i})$ we guess the order of the edges appearing, which is at most $ |E_{\pi_{s}}|! |E_{\pi_{s}}|^{\sigma_{C}} $ guessings. Note at each recursion it may increase at most $|E_{\pi_{s}}|$ many the number of start-end pairs and the depth of the recursion is $\delta$. Thus $\sigma_{C}\leq \delta |E_{\pi_{s}}|$ and 
$|E_{\pi_{s}}|! |E_{\pi_{s}}|^{\sigma_{c}} \leq (\frac{16\kappa' \delta}{\epsilon} )^{2\kappa} ! (\frac{16\kappa' \delta}{\epsilon} )^{2\kappa \delta ( \frac{16\kappa' \delta}{\epsilon} )^{2\kappa}  }   \leq n^{O((\frac{\delta}{\eps})^{2\kappa+1})}$.
 
We show the size of the dynamic programming table is at most  $n^{O((\frac{\delta}{\epsilon})^{2\kappa +1})}$.
Recall an entry of the table is of the form $A[C,k_{C},(s_{i},t_{i})_{i=1}^{\sigma_{C}}]$. For $C$, there are at most $(3\log n 2^{\kappa} )^{\delta}  $ cluster nodes in $\Gamma$ because the size of $\Gamma$ is at most $(3\log n 2^{\kappa} )^{\delta} $. For $k_{C}$, there are at most $n$ possible value of $k_{C}$ to consider. For $\{s_{i},t_{i} \}_{i=1}^{\sigma_{C}}$, there are at most $ n^{2\sigma_{C} } $ start-end pairs to consider, which is at most $ n^{2\delta (\frac{16\kappa'\delta}{\epsilon})^{2\kappa} } $ because $\sigma_{C}$ is at most $\delta |E_{\pi_{s}}| $. 

Therefore, computing the whole DP table and finding the near optimum path $P'$ as in Theorem \ref{thm:muexcess} takes at most $ n^{ O( (\frac{\delta}{\epsilon})^{2\kappa+1} )} $ time.

\end{proof}

The goal is to compute $A[c_{0},k,(s,t)]$ where $k$ and $(s,t)$ are specified in the {$k$-stroll} instance. Proof of Theorem \ref{thm:kstrollDBL} follows from this lemma and Theorems \ref{thm:structure} and \ref{thm:muexcess}.

%%%%%%%%%%%%%%%%%%%%%%%%%%%%%%%%%%%%%%%%%%%%%%%%%%%%%%%%%%%%%%
\subsection{Approximating P2P Orienteering on Doubling Metrics}\label{sec:p2p}
In this section we prove Theorem \ref{thm:P2PDBL} using the results of previous section for {$k$-stroll}.
%The following lemma is the key to prove Theorem \ref{thm:P2PDBL}.

\begin{lemma}\label{lem:p2porienteering}
    Let $G=(V,E)$ be graph with constant doubling dimension $\kappa$ and given an instance of a P2P orienteering on $G$ with specified budget $B$, start node $s\in V$ and end node $t\in V$. Let $P^*$ be the optimal for this instance and $k=|P^*|$. Then we can get a $s$-$t$ path that visits at least $(1-\epsilon)k$ vertices in $G$ with the length at most $B$.
\end{lemma}

\begin{proof}
The idea of the proof is starting from $P^*$ we can show there exists a subpath of it $P'$, where $|P'|\geq (1-\eps)k$ and $||P'||$ is smaller
than $B$ by at least $\eps$ times the $\mu$-excess of $P^*$ for $\mu=\Theta(1/\eps)$. So if we find
a $(\eps,\mu)$-approximation for {$k$-stroll} with parameter $(1-\eps)k$ and $\mu=\Theta(1/\eps)$ it's length will respect the budget bound $B$. To do that we break $P^*$ into $\mu$ equal size segments each having $\eps\cdot k$  nodes and short-cut the longest segment.

Let $\mu=\lceil \frac{1}{\eps}\rceil+1$ and assume that $k\geq\mu^2$
(otherwise we can find $P^*$ in polynomial time by exhaustive search in time $O(n^{1/\eps^2})$).
We construct a subsequence of $1,\cdots,k$ to define a $\mu$-jump of $P^*$: Let $a_{i}=\lceil \frac{(i-1)(k-1)}{\mu-1} \rceil+1$, $1\leq i\leq \mu$. Note that $a_{1}=1$ and $a_{\mu}=k$. Let $P^*_{1},P^*_{2},\cdots,P^*_{\mu-1}$ be the subpaths of $P^*$ divided by $\{v_{a_{1}},\cdots,v_{a_{\mu}}\}$, i.e. $P^*_{i} $ is a subpath of $P^*$ with the source $v_{a_{i}}$ and sink $v_{a_{i+1}}$.

  For each $P^*_{i}$ we consider the $2$-excess of it and let $P^*_{j}$ be the subpath with maximum $2$-excess among $\{P^*_{i}\}_{i=1}^{\mu-1}$, i.e. $j=\arg\max_{i}\calE_{P^*_{i},2}$. Note $|P^*_{j}|=a_{j+1}-a_{j}+1= (\lceil \frac{j(k-1)}{\mu-1} \rceil+1 ) -(\lceil \frac{(j-1)(k-1)}{\mu-1} \rceil+1 )+1\leq \lceil  \frac{k-1}{\mu-1} \rceil +1.$
Then let $P'$ be the path exactly the same as $P^*$ except $P'$ directly connects $v_{a_{j}}$ and $v_{a_{j+1}}$ in $P^*_{j}$. From the construction of $P'$, 

\begin{eqnarray}
|P'|=k-|P_{j}|+2 \geq k-\lceil  \frac{k-1}{\mu-1} \rceil-1+2 \geq k-\lfloor \frac{k-1}{\mu-1} \rfloor \geq (1-\epsilon)k
\quad\quad \text{since $k\geq\mu^2$, and}\label{eqn:eq1}\\
||P'||=||P^*||-\calE_{P^*_{j},2}=B-\calE_{P^*_{j},2}.\label{eqn:eq2}
\end{eqnarray}
We consider $P'$ as a feasible solution for a $k'$-stroll instance with $k'=|P'| \geq (1-2\epsilon)k$ and $s,t\in V$. By Theorem \ref{thm:muexcess}, we can compute a $(\epsilon,\mu)$-approximation for this instance where $\mu=\lceil \frac{1}{\epsilon} \rceil +1$, denoted as $P''$:
\begin{eqnarray}
|P''|\geq k'\geq (1-\epsilon)k \quad\quad \text{and} \label{eqn:eq3}\\
||P''||\leq ||P'||+ \epsilon \calE_{P',\mu} \leq B-\calE_{P^*_{j},2} + \epsilon \calE_{P',\mu} \leq B-\calE_{P^*_{j},2} + \frac{1}{\mu-1} \calE_{P',\mu},\label{eqn:eq4}
\end{eqnarray}
where the 2nd line uses (\ref{eqn:eq2}).
We consider the $\mu$-jump of $P'$ defined by $ \langle v_{a_{1}},\cdots,v_{a_{\mu}}\rangle $. Note that this is is also a $w$-jump of $P^*$. By the definition of excess and how we obtained $P'$ from $P^*$ (by short-cutting $P^*_{j,2}$):
\begin{equation}\label{eqn:eq5}
\calE_{P',\mu} \leq ||P'||- ||\langle v_{a_{1}},\cdots,v_{a_{\mu}}\rangle  \rangle || 
=\sum_{i=1}^{\mu-1}\calE_{P^*_{i},2} -\calE_{P^*_{j},2}.
\end{equation}

Thus, using (\ref{eqn:eq4}) and (\ref{eqn:eq5}):
%\begin{equation*}
$||P''||\leq B-\calE_{P^*_{j},2}+\frac{1}{\mu-1}(\sum_{i=1}^{\mu-1}\calE_{P^*_{i},2} -\calE_{P^*_{j},2} ) \leq B$,
%\end{equation*}
where the last inequality follows from the fact that $\calE_{P^*_j,2}$ is the largest among all indices.

%By replacing $\eps$ with $\eps/2$ in the above we find  $|P''|\geq (1-\eps)k$ and $||P''||\leq B$ as wanted.
\end{proof}

However, for the P2P orienteering instance, $P$ and $k$ in Lemma \ref{lem:p2porienteering} are unknown in advance. Therefore, we will consider all possible integers $1\leq k\leq n$ and for each $k$ we get the approximation for {$k$-stroll} on $G$ with specified $k$ and $s,t\in V$. We return the maximum $k$ such that the length of path we get for $k$-stroll is at most $B$.
This completes the proof of Theorem \ref{thm:P2PDBL}.
%\begin{theorem}
%    Given a graph $G=(V,E)$ with constant doubling dimension $\kappa$, an instance of P2P orienteering on $G$ with specified budget $B$, start node $s\in V$ and end node $t\in V$. We can get a $(1+\epsilon)$ approximation for this instance in quasi polynomial time.
%\end{theorem}

%%%%%%%%%%%%%%%%%%%%%%%%%%%%%%%%%%%%%%%%%%%%%%%%%%%%%%%%%%%%%%%%%%%%%%%
\section{An Approximation Scheme for Deadline TSP in Doubling Metrics}\label{sec:qptas_deadlineTSP}

In this section we prove Theorem \ref{thm:deadlineDBL}. Our algorithm builds upon ideas from \cite{FriggstadS22} for $O(1)$-approximation for deadline TSP for general metrics combined with new ideas as well as those developed in the previous section to get an approximation scheme for deadline TSP on doubling metrics.
Friggstad and Swamyy \cite{FriggstadS22} present an $O(1)$-approximation for deadline TSP running in time $n^{O(\log n\Delta)}$ assuming that all distances are integers and at least 1. They use the notion of regret which is the same as $2$-excess. Note if path $P$ visits $u$ and then $v$ then
short-cutting the subpath $P_{uv}$ (replacing it with edge $uv$) will save a length which is exactly $\calE_{P_{uv},2}$. They guess a sequence of vertices 
$v_{0}=s,v_{1},\cdots,v_{\ell}$ of an optimal solution (with $\ell=O(\log\Delta)$) such that the $2$-excess of the subpaths of optimal increase geometrically: $\calE_{P_{v_{i}v_{i+1}},2}\geq \alpha^{i}$, where $\alpha$ is some constant satisfying $1+\alpha\geq \alpha^{2}$. They consider a set of {P2P orienteering} instances with start node $v_{i}$, end node $v_{i+1}$, and length budget  $d(v_{i},v_{i+1})+\alpha^{i}$. These instances are not independent however, hence this is a more general problem that we call multi-path orienteering with group budget constraints (see \cite{FriggstadS22}).
They show given an $\beta$-approximation for {P2P orienteering}, at an another $O(1)$-factor loss, one can turn it into an $O(\beta)$-approximation for multi-group multi-path orienteering. Then they concatenate these paths; these paths are not respecting the deadlines however. In order to make them deadline respecting, from every three consecutive paths they shortcut two of them, so another $O(1)$-factor loss. The saving for the shortcutting of two paths is enough for the deadline of every vertex in the third path being satisfied. To obtain $O(1)$-approximation for {P2P orienteering} instances with groups, they use known reductions from the problem of maximum coverage with group budgets to classic maximum coverage and use algorithm of \cite{chekuri2004maximum} to get a constant approximation via a reduction to classic {P2P orienteering} (see \cite{FriggstadS22} for more details). Putting everything together, to obtain an $O(1)$-approximation for deadline TSP they lose $O(1)$ factor in three steps. 
%First, starting from $O(1)$-approximation for general {P2P orienteering} one loses another $O(1)$ factor to get an approximation for multi-group multi-path orienteering. To combine the solutions and convert it to a feasible solution of the deadline TSP instance, they lose another $O(1)$ factor. 

In our setting in order to get a $(1+\eps)$-approximation for deadline TSP on graphs with bounded doubling dimension, we have to change all these steps so that we don't lose more than $(1+\epsilon)$ factor in any step. It turns out  it becomes significantly more complex to maintain an at most $(1+\epsilon)$ loss in any step. As in \cite{FriggstadS22}, we assume distances are integer and $\geq 1$ (this can be done as pointed out in \cite{FriggstadS22} by  scaling). 

{\bf Overview of the proof:}
Suppose we have guessed a sequence of vertices $\langle v_{0}=s,v_{1},\cdots,v_{m}\rangle$ of an optimum solution $P^*$ where the $\mu$-excess of the sub-path $P^*_{v_i,v_{i+1}}$ is (at least) $\alpha^i$, where $\alpha=1+\eps$ (for simplicity assume the increase in
excess is exactly $\alpha^i$). We also assume we have guessed
the lengths of these sub-paths, say $||P^*_{v_i,v_{i+1}}||=B_i$. 
Note that the vertices visited in $P^*_{v_i,v_{i+1}}$ all must
have a deadline at least as big as the visit time of $v_i$ in $P^*$(which
we have guessed since we know all previous $B_j$'s, $j<i$); let $W_i$
denote this set of vertices.
Let $\calI_i$ be the {P2P orienteering} instance with start-end pair $v_i,v_{i+1}$, budget $B_i$, where the vertices allowed to visit are $W_i$.
Note that the sub-paths $P^*_{v_i,v_{i+1}}$ form a solution to these instances $\calI_i$ for $i\geq 1$.
If we have $v_i$'s and $B_i$'s, we can try to solve these $O(\log \Delta)$ instances simultaneously
(our DP described for {P2P orienteering} can be expanded to handle when we have $O(\log \Delta)$ instances to be solved on the same ground set simultaneously). 

One problem is that the vertices visited in $\calI_i$ might be violating their deadlines slightly. We will show that this violation will be small (using the assumption that the excess of subpath $P^*_{v_i,v_{i+1}}$ was $\alpha^i$) and that by using a similar technique as we did to convert an approximation for {$k$-stroll} to an approximation for {P2P orienteering} we can drop a small fraction of vertices visited in all instances such that the total saving in time achieved for all $\calI_j$ with $j<i$ is enough to ensure all the (remaining) vertices in $\calI_i$ are visited before their deadline. Now we start explaining the details of the proof.

Let $G = (V, E)$ be a graph with constant doubling dimension $\kappa$, given a start node $s\in V$ and deadline $D(v)$ for all $v\in V$ as an instance $\calI$ of deadline TSP on $G$. Suppose $\mu$ is a constant (will be fixed to be $\lfloor\frac{1}{\eps}\rfloor+1$ later and let $\alpha=(1+\eps)$. Let $P^*$ be an optimum solution for this instance and $\langle v_{0},v_{1},\cdots,v_{m}\rangle$ be a sequence of vertices in $P^{*}$ satisfying the following properties:

\begin{itemize}
    \item $v_{0}=s$ is the start node of $P^{*}$.
    \item  $v_{i+1}$ is the first vertex in $P^{*}$ after $v_{i}$ with $ \calE_{P^*_{v_{i}v_{i+1}},\mu}> \alpha^{i}$, except possibly for $v_m$ (the last vertex of $P^{*}$).
\end{itemize}

So each vertex $v_{i+1}$ is the first vertex along the optimum
after $v_i$ such that the $\mu$-excess of the subpath $P^*_{v_iv_{i+1}}$ is at least $\alpha^i$. We will be guessing these vertices $v_i$'s eventually and try to find (good) approximations of $P^*_{v_iv_{i+1}}$. We can assume that $|P^*_{v_iv_{i+1}}|\geq \mu^2$, otherwise we can compute $P^*_{v_iv_{i+1}}$ exactly using exhaustive search.
We also denote the vertex on $P^*$ immediately before $v_{i+1}$ by
$v'_i$. Note  $||P^{*}|| \leq n\Delta_{G}$, thus  $m\leq h\delta$ (where $\delta=\log\Delta$) for some constant $h=h(\eps)>0$. 
For each $0\leq i <m$, we break $P^*_{v_iv_{i+1}}$ into $\mu-1 $ subpaths of (almost) equal sizes, denoted as $P^{*}_{i,j},1\leq j<\mu$, by selecting a $\mu$-jump $J_i: v_i=u^1_i,u^2_i,u^3_i,\ldots,u^{\mu-1}_i,u^{\mu}_i=v_{i+1}$ as follows. Assume $P^*_{v_iv_{i+1}}=\langle v_{i,1},\ldots,v_{i,k_i}\rangle$ where $v_i=v_{i,1}$ and $v_{i,k_i}=v_{i+1}$, let $a_{j}= \lceil \frac {(j-1)(k_{i}-1)  }{\mu -1} \rceil +1$ then $a_{1}=1$, $a_{\mu}=k_{i}$, and if we consider  $v_{i,a_{1}},\cdots,v_{i,a_{\mu}}$ then we obtain $J_i$ by letting $v_{i,a_j}=u^j_i$. Suppose $J^*_\mu=J^*_\mu(P^*_{v_iv_{i+1}})$ is the optimum $\mu$-jump of $P^*_{v_iv_{i+1}}$, which is the $\mu$-jump with the maximum length.
Recall that $\calE_{P^*_{v_iv_{i+1}},\mu}$ is the $\mu$-excess of $P^*_{v_iv_{i+1}}$ and with  
$B_i= ||J^*_\mu(P^*_{v_iv_{i+1}})||+\calE_{P^*_{v_iv_{i+1}},\mu}$ we have $||P^*_{v_iv_{i+1}}||= B_i$.To simplify the notation we denote the $\mu$-excess of subpath $P^*_{v_iv_{i+1}}$ by $\calE^*_{i,\mu}$. We also use $\calE'_{i,\mu}$ to denote the $\mu$-excess of path $P^*_{v_iv'_i}$. From the definition of $\langle v_{0},v_{1},\cdots,v_{m}\rangle$ it follows that $\calE^*_{i,\mu}>\alpha^i$ and $\calE'_{i,\mu}\leq \alpha^i-1$.
Let $v$ be an arbitrary vertex in $P^*_{v_iv_{i+1}}$ that falls in between 
$u^j_i$ and $u^{j+1}_i$. We use $||J_i(v_i,u^{j}_i)||$ to denote the length of $J_i$ from $v_i$ to $u^j_i$ (i.e. following along $J_i$ from the start node $v_i$ to $u^j_i$). Define $L_{i,j}=\sum_{j=0}^{i-1} ||P^*_{v_j v_{j+1}}||+||J_i(v_i,u^{j}_i)||$. 
Note that the visiting time of $v$ in $P^*$ (and hence the deadline of $v$) is lower bounded by $L_{i,j}$. 

Let $N_{i,j}=\{ v: D(v)\geq L_{i,j} \}$. Observe that if we consider $P^*_{v_iv_{i+1}}$ broken up into several legs
$P^{*}_{u_{i}^{j}u_{i}^{j+1}}$, $1\leq j<\mu$, then it is a {P2P orienteering} instance with start node $v_i$, end node $v_{i+1}$ and given extra intermediate nodes $u^j_i$ and the path is supposed to go through these intermediate nodes in this order; so it consists of $\mu-1$ {\em legs} where leg $j$ is between ${u^j_i,u^{j+1}_i}$ and uses vertices in $N_{i,j}\subset V$  and total budget $B_i=||J^*_\mu(P^*_{v_iv_{i+1}})||+\calE^*_{i,\mu}$.
%We consider for all $0\leq i\leq m$ and all $1\leq j \leq \mu$ concurrently:

\begin{definition}\label{def:mglo} (multiple-groups-legs orienteering)
    Let $G=(V,E)$ be a graph, given $m$ groups each with $\mu-1$ legs, where
    leg $\ell$ of group $i$ ($0\leq i<m$, $1\leq \ell<\mu$) has 
    start and end node pair $(s_{i,\ell}, t_{i,\ell})$ and can use vertices from $N_{i,\ell}\subseteq V$ (we have
    the property that end-node of leg $\ell$ is the same as start node of leg $\ell+1$: $t_{i,\ell}=s_{i,\ell+1}$),
    total budget for all the legs of group $i$ is $B_{i}$. The goal is to find a collection of paths $Q_{i,\ell}$, for $0\leq i<m$, $1\leq \ell\leq \mu -1$, such that $Q_{i,\ell}$ is a $s_{i,\ell}$-$t_{i,\ell}$ path (and so concatenation of all legs of group $i$ gives a single path from start node of the first leg to the last node of leg $\mu-1$), such that $\sum_{\ell=1}^{\mu-1}||Q_{i,\ell}||\leq B_{i} $ and $| \cup_{i=0}^{m-1} \cup_{\ell=1}^{\mu-1} (Q_{i,\ell}\cap N_{i,\ell}) |$ is maximized.
\end{definition}

Note $P^{*}_{u_i^{j}, u_i^{j+1}}$ ($0\leq i<m, 1\leq j<\mu$) is a feasible solution of the multiple groups-legs orienteering instance with groups $ 0\leq i<m$ and legs $1\leq j<\mu$, start and end node pairs $(u_i^{j}, u_i^{j+1})$, budgets $B_{i}$ and subset $N_{i,j}$.
Consider the instance restricted to group $i$ alone an instance of multiple-leg {P2P orienteering} where we have to find a $v_iv_{i+1}$-path that goes through all $u^j_i$'s in order (so has $\mu-1$ legs) and has budget $B_i$; call this instance $\calI_i$. Note $P^{*}_{v_{i}v_{i+1}}$ is a feasible solution to $\calI_i$, thus the optimal of $\calI_i$ visits at least $|P^{*}_{v_{i}v_{i+1}}|$ many vertices in $\cup_{j=1}^{\mu-1} N_{i,j}$. We consider all $\calI_i$'s concurrently, i.e. an  instance with pairs of start-end nodes $v_iv_{i+1}$ for group $i$ where each group is also required to go through vertices of $J_i$ in that order and thus has $\mu-1$ legs where each leg has start-end pair $u^j_i,u^{j+1}_i$, total budget $B_{i}$ for all the legs of group $i$, and the subset $N_{i,j}\subset V$, $0\leq i< m$, $1\leq j<\mu$ for leg $j$ of group $i$. For a path $Q$, let $Q\cap N_{i,j}$ be the set of vertices $Q$ visited in $ N_{i,j}$. Using an argument similar to that of Theorem \ref{thm:structure}
we can show there is a $(1+\epsilon)$-approximation i.e. a set of paths $Q'_{i,j}$ such that $Q'_{i,j}$ is a  $u^j_{i},u^{j+1}_i$-path and if we define concatenation of different legs of group $i$ by $Q'_i=Q'_{i,1}+\ldots+Q'_{i,\mu-1}$ then:
\begin{equation}\label{eqn12}
    ||Q'_i||=\sum_{j=1}^{\mu-1}||Q'_{i,j}||\leq ||P^*_{v_iv_{i+1}}||-\eps\calE^*_{i,\mu}\quad\quad\text{and}\quad\quad
|\bigcup_{i=0}^{m-1} Q'_i|=|\bigcup_{i=0}^{m-1}\cup_{j=1}^{\mu-1}(Q'_{i,j}\cap N_{i,j})|\geq (1-4\eps)|P^*|.
\end{equation} 

We also show that if $v$ is visited by $Q'_{i,j}$ then if $Q'_{i,j}(u^j_i,v)$ denotes the segment of path $Q'_{i,j}$ from $u^j_i$ to $v$, then the length of the segment from $v_i$ to $v$ in $Q'_i$ can be upper bounded:

\begin{equation}\label{eqn13}
||Q'_i(v_i,v)||=\sum_{\ell=1}^{j-1}||Q'_{i,\ell}|| + ||Q'_{i,j}(u^j_i,v)|| \leq ||J_i(v_i,u^j_i)||+(1-\eps)\calE'_{i,\mu}.
\end{equation}

We will prove the existence of such paths $Q'_{i,j}$ with a structure in the next theorem and also how to find these using a DP.
For now suppose we have found such paths $Q'_i$ as described above. We concatenate all these paths to obtain the final answer $\calQ=Q'_0+Q'_1+\ldots+Q'_{m-1}$. We show the vertices of $\calQ$ 
are visited before their deadlines and hence we have an approximation  for deadline TSP. Given the bounds given for the sizes of $Q'_i$'s in (\ref{eqn12}),  the number of vertices visited overall (respecting their deadlines) is at least $(1-4\eps)|P^*|$.

To see why the vertices in $\calQ$ respect their deadlines consider an arbitrary node $v\in Q'_i$.
Note that each $Q'_i$ contains the vertices in $J_i$ (as those are the vertices that define $\mu-1$ legs of the $i$'th group). Suppose $v$ is visited in $Q'_{i,j}$, i.e. between $u^j_i$ and $u^{j+1}_i$. Therefore, the visit time of $v$ in $\calQ$, i.e. $||\calQ_{sv}||$ is bounded by:

\begin{eqnarray*}
  ||\calQ_{sv}|| &=& \sum_{\ell=0}^{i-1}||Q'_\ell|| + ||Q'_i(v_i,v)|| \\
    &\leq& \sum_{\ell=0}^{i-1} (||P^*_{v_\ell v_{\ell+1}}|| - \eps\calE^*_{\ell,\mu}) + ||J_i(v_i,u^j_i)||+(1-\eps)\calE'_{i,\mu} \quad\quad\quad\quad\quad\mbox{using (\ref{eqn12}) and (\ref{eqn13})}\\ [-10pt]
    &=& L_{i,j} + (1-\eps)\calE'_{i,\mu}-\eps\sum_{\ell=0}^{i-1}\calE^*_{\ell,\mu}\leq D(v),
\end{eqnarray*}

\noindent where the last inequality follows from the fact that $\calE'_{i,\mu}\leq\alpha^i-1$ and
$\calE^*_{\ell,\mu}> \alpha^\ell$ so $\eps\sum_{\ell=0}^{i-1}\calE^*_{\ell,\mu}\geq  \eps\sum_{\ell=0}^{i-1}\alpha^\ell = (\alpha^i-1)\geq \calE'_{i,\mu}$.

So we need to show the existence of $Q'_i$ as described and how to find them.
The road to get these paths $Q'_i$ is an extension of what we had for {multi-path orienteering} in the previous section and is obtained by a DP that computes multi-group-legs with multiple paths (i.e. instances $\calI_i$'s) concurrently. The proof of following theorem is an extension of that of Theorems \ref{thm:structure} and \ref{thm:muexcess}.

\begin{theorem}(multi-groups-legs multi-paths orienteering structure theorem )\label{the:mppstructure}
Let $G=(V,E)$ be a graph with constant doubling dimension $\kappa$, given $s\in V$ and $D(v)$ for all $v$ as an instance of deadline TSP and $P^{*}$ be an optimal solution. Let $v_{i}$ ($0\leq i<m$), $u_{i}^{j}$ ($0\leq i<m$, $1\leq j <\mu$), $B_{i}$, and $N_{i,j}$ as described above for $\mu=\lfloor\frac{1}{\eps}\rfloor+1$. Consider a multi-groups-legs orienteering instance with groups $ 0\leq i<m$ and legs $1\leq j<\mu$, start and end node pairs $(u_i^{j}, u_i^{j+1})$, budgets $B_{i}$, and subset $N_{i,j}$. then we can construct a $\gamma$-split-tree (with $\gamma=n^{3h\delta}$) such that with probability at least $1-\frac{1}{n}$ there exists a determining function $\Phi$ and a corresponding split-tree hierarchical decomposition $T=\Gamma|_{\Phi}$ of $G$ and a structured solution of the multi-groups-legs orienteering instance $Q'_{i,j }$, ($0\leq i<m$, $1\leq j<\mu$) such that if we define $Q'_i=Q'_{i,1}+\ldots+Q'_{i,\mu-1}$ for each $0\leq i<m$, then:
\begin{enumerate}
   \item $|\cup_{i=0}^{m-1}\cup_{j=1}^{\mu-1} (Q'_{i,j}\cap N_{i,j})|\geq  (1-4\eps)|\cup_{i=0}^{m-1}\cup_{j=1}^{\mu-1} P^{*}_{ u_i^{j}, u_i^{j+1}}|=(1-4\eps)|P^{*}|$.
   \item $||Q'_i||=\sum_{j=1}^{\mu-1}||Q'_{i,j }||\leq B_i-\eps\calE^*_{i,\mu} = ||J^*_{\mu}(P^*_{v_i,v_{i+1}})|| + (1-\eps)\calE^*_{_i,\mu}$
   and for any vertex visited by $Q'_i$, say $v$ visited in $Q'_{i,j}$ we have the length of the path from $v_i$ to $v$ in $Q'_i$, denoted by $||Q'_i(v_i,v)||$, satisfies: 
   $||Q'_i(v_i,v)||\leq ||J^*_\mu(P^*_{v_i,v'_i})||+(1-\eps)\calE'_{i,\mu}$.
   
   \item For any cluster $C\in T$ and a partition $\pi_{\Phi(C)}$ of $C$ defining its children in $T$, let $R_{\pi_{\Phi(C)}}, R_{\pi_{\Phi(C)}}$ are the set of bridge-edges and portal-edges of $C$ cut by the partition, we have either:
   \begin{itemize}
      \item  $ |Q^{'}_{i} \cap R_{ \pi_{\Phi(C)} } |\leq 2\kappa'\frac{\log n}{\epsilon}  $ or
      \item $|Q^{'}_{i} \cap R'_{ \pi_{\Phi(C) }}| \leq (\frac{16\kappa' \delta}{\epsilon})^{2\kappa}$.
   \end{itemize}
\end{enumerate}
\end{theorem}

\begin{proof}

To simplify  notation, for each $i$, let $P^{*}_{i}$ be the subpath $P^*_{v_iv_{i+1}}$ and $P^{*}_{i,j} $ be the subpath $P^*_{u_i^j u_{i}^{j+1}}$. For each $P^{*}_{i,j}$ we consider the $2$-excess of it and let $j',j''$
be the two indices that $P^*_{i,j'},P^*_{i,j''}$ have the largest two 2-excess among indices $1,\ldots,\mu-2$.
We consider short-cutting the three subpaths $P^*_{i,j'},P^*_{i,j''},$ and $P^*_{i,\mu-1}$ and let the resulting
path obtained from $P^*_i$ by these short-cutting be $Q_i$. In other words, $Q_i$ is the same as $P^*_i$ except that each of $P^*_{i,j'}$, $P^*_{i,j''}$, and $P^*_{i,\mu-1}$ are replaced with the direct edges $(u^{j'}_i,u^{j'+1}_i)$, $(u^{j''}_i,u^{j''+1}_i)$, and $(u^{\mu-1}_i,t_i)$, respectively.
Let $D^2_i=\calE_{P^*_{i,j'},2}+\calE_{P^*_{i,j''},2}$ and $D^3_i=D^2_i+\calE_{P^*_{i,\mu-1},2}$. We have $D^2_i\geq\frac{2}{\mu-1}\sum_{j=1}^{\mu-2}\calE_{P^*_{i,j},2}$. Recall that $v'_i$ is the vertex on $P^*$ just before $v_{i+1}$, so $||P^*_i(v_i,v'_i)||=||J^*_\mu(P^*_i(v_i,v'_i))||+\calE'_{i,\mu}$ (this is from the definition of optimum $\mu$-jump of the path $P^*_i(v_i,v'_i)$). So for vertices in $Q_{i}$ we can bound the length of the subpath from $v_i$ to $v$ by: 

\begin{equation}\label{eqn10}
||Q_i(v_i,v)||\leq ||P^*_i(v_i,v)||-D^2_i\leq ||J^*_\mu(P^*_i(v_i,v'_i))||+\calE'_{i,\mu}-\frac{2}{\mu-1}\sum_{j=1}^{\mu-2}\calE_{P^*_{i,j},2}.
\end{equation}

Also for the total length of $Q_i$ we have:
\begin{equation}\label{eqn11}
    ||Q_i||\leq ||P^*_i|| - D^3 \leq ||P^*_i||-\frac{2}{\mu-1}\sum_{j=1}^{\mu-1}\calE_{P^*_{i,j},2}.
\end{equation}

As for the size note that $|Q^{*}_{i,j'}|=a_{j'+1}-a_{j'}+1 =(\lceil \frac {(j'+1)(k_{i}-1)  }{\mu -1} \rceil +1) -  (\lceil \frac{j'(k_{i}-1)  }{\mu-1}  \rceil+1) +  1 \leq \lceil\frac{k_{i}-1}{\mu-1}  \rceil +1 $. Same bound holds for $|P^*_{i,j''}|$ and $|P^*_{i,\mu-1}|$.
From the construction of $Q_{i}$: $|Q'_{i}|=k_{i}-|Q^{*}_{i,j'}|+2 -|P^*_{i,j''}|+2 -|P^*_{i,\mu-1}| +2 \geq k_{i} - 3\lceil\frac{k_{i}-1}{\mu-1}  \rceil +6 \geq  k_{i} - 3\lfloor\frac{k_{i}-1}{\mu-1} \rfloor  \geq (1-4\epsilon) k_{i} $ since we assumed $k_i\geq\mu^2$.

Suppose $\Gamma$ is a $\gamma$-split-tree for $G$. We build  $Q^{'}_{i,j}$, $0\leq i<m$, $1\leq j<\mu$ iteratively based on $Q_{i,j}$ and at the same time build $\Phi(\cdot)$ and hence $T$ from the top to bottom.
Initially we set $Q^{'}_{i,j}$ to be $Q_{i,j}$ for $0\leq i<m$,$1\leq j<\mu$, and start from the root cluster node $C_{0}$ of $T$. At any point when we are at a cluster node $C$ we consider $\cup_{j=1}^{\mu-1} Q'_{i,j}$ with respect to $C$ and whether it is dense or sparse.

For any group $i$ that $\cup_{j=1}^{\mu-1} Q^{'}_{i,j}$ is sparse with respect to $C$, we don't modify $(\cup_{j=1}^{\mu-1} Q^{'}_{i,j})\cap C$ when going down from $C$ to any split node $s$ of $C$. Consider any child split nodes $s$ of $C$ and let $\pi_{s}$ be the partition of $C$ according to the split node $s$. For each $i,j$ let $E^{i,j}_{\pi_s}$ be the bridge-edges of $Q'_{i,j}$ with respect to $\pi_s$. Consider the number of edges of $(\cup_{j=1}^{\mu-1}Q^{'}_{i,j})\cap C$ that are bridge-edges based on this partition, i.e. $| (\cup_{j=1}^{\mu-1}Q^{'}_{i,j}) \cap R_{\pi_{s}}| =|\cup_{j=1}^{\mu-1}E^{i,j}_{\pi_{s}}| $. Using Lemma \ref{lem:sparse}: $\E( |\cup_{j=1}^{\mu-1} E^{i,j}_{\pi_{s}}|) \leq \kappa' \frac{\log n}{\epsilon}  $. Let the event $\Lambda_{i,s}$ be the event that $ |\cup_{j=1}^{\mu-1}E^{i,j}_{\pi_{s}}|\leq 2\kappa'\frac{\log n}{\epsilon} $. By Markov inequality $\Pr[ \Lambda_{i,s}   ]\geq \frac{1}{2}$. Recall there are $\gamma=n^{3h\delta}=2^{3h\delta\log n}$ many children split nodes of $C$, i.e. $\gamma$ random partition of $C$ and (since $P^*$ was broken into $m$ subpaths $P^*_i$) there are at most $m=h\delta$  many paths $Q^{'}_{i}$ that are sparse with respect to $C$. Thus the probability that for at least one child split node  $s$ of $C$, events $\Lambda_{i,s}  $ holds for all $Q^{'}_{i}$'s that are sparse with respect to $C$ is at least 
$ 1- [1- (\frac{1}{2})^{h\delta}]^{\gamma} \geq 1- [1-\frac{1}{2^{h\delta}}]^{2^{3h\delta\log n} } \geq  1-\frac{1}{n^{3}}$.

Now consider any group $i$ such that $\cup_{j=1}^{\mu-1} Q^{'}_{i,j} $ is $\eta$-dense with respect to $C$ and let $s$ be an arbitrary child split nodes of $C$. We will modify $\cup_{j=1}^{\mu-1}Q^{'}_{i,j} $ going down the split node $s$ to be portal respecting with respect to $\pi_{s}$ as described in Lemma \ref{lem:dense}. Note the set of edges of $\cup_{j=1}^{\mu-1}Q^{'}_{i,j} $ crossing $\pi_{s}$ after making it portal respecting with respect to $\pi_{s}$ is a subset of $R'_{\pi_{s}}$. Thus $|(\cup_{j=1}^{\mu-1} Q^{'}_{i,j})\cap R'_{\pi_{s}}|\leq R'_{\pi_{s}} \leq ( \frac{16\kappa' \delta}{\epsilon} )^{2\kappa}$ in this case. According to Lemma \ref{lem:dense}, the expected increase of length of  $\cup_{j=1}^{\mu-1} Q^{'}_{i,j} $ after making it portal respecting with respect to $\pi_{s}$ is at most $ \frac{\epsilon}{2\delta} ||(\cup_{j=1}^{\mu-1} Q^{'}_{i,j}) \cap C||$. Let $\Lambda^{'}_{i,s}$ be the event that the increase of length of $\cup_{j=1}^{\mu-1} Q^{'}_{i,j} $ after making it portal respecting with respect to $\pi_{s}$ is at most $\frac{\epsilon}{\delta}||(\cup_{j=1}^{\mu-1} Q^{'}_{i,j} )\cap C|| $. By Markov inequality,  $\Pr[ \Lambda'_{i,s} ] \geq \frac{1}{2} $.  Recall there are $\gamma$ many children split nodes of $C$ and there are at most $m=h\delta$ many paths $Q^{'}_{i}$ that are $\eta$-dense with respect to $C$. Thus the probability that for at least one child split node  $s$ of $C$, events $\Lambda^{i}_{s}$ happens for all $Q^{'}_{i}$'s that are $\eta$-dense with respect to $C$ is at least 
$ 1- [1- (\frac{1}{2})^{h\delta}]^{\gamma} \geq 1- [1-\frac{1}{2^{h\delta}}]^{2^{3h\delta\log n} } \geq  1-\frac{1}{n^{3}}$.

Therefore, such $s$ exists for cluster $C$ with the probability at least $1-\frac{2}{n^{3}}$ for all $i$ and we can define $\Phi(C)$. Once we have $\Phi(\cdot)$ defined for clusters at a level of $\Gamma$, we have determined the clusters at the same level of $T=\Gamma|_{\Phi}$. Note there are at most $n$  cluster nodes in one level of $T$. Thus with probability at least $1-\frac{2}{n^{2}}$ such split nodes exist for all cluster nodes in one level. Since the height of $\Gamma$ is $\delta$, thus with probability at least   $ (1-\frac{2}{n^{2}})^{\delta}\geq 1-\frac{1}{n}$ (assuming that $\delta$ is polylogarithmic in $n$ )  such $\Phi(\cdot)$ is well defined over all levels.

Note $ Q^{'}_{i,j} $ visits all the of vertices from $ N_{i,j}$ that $Q_i$ was visiting, so  
$$|\cup_{i=0}^{m-1}\cup_{j=1}^{\mu-1}(Q'_{i,j}\cap N_{i,j})|\geq (1-4\eps)|\cup_{i=0}^{m-1}\cup_{j=1}^{\mu-1} P^{*}_{u_i^{j}, u_i^{j+1}}|=(1-4\eps)|P^{*}|.$$ For any cluster $C\in T$, the increase in length of $Q'_i=\cup_{j=1}^{\mu-1} Q^{'}_{i,j}$ only stems from when it is $\eta$-dense with respect to $C$. 
Using Theorem \ref{thm:muexcess}, $Q'_i=\cup_{j=1}^{\mu-1} Q^{'}_{i,j}$ is a $(\epsilon,\mu)$-approximation of $Q_{i}$. Thus, combined with (\ref{eqn11}):

\begin{equation}\label{eqn5}
||Q^{'}_{i} || \leq ||Q_{i}||+\epsilon \calE_{Q_{i} ,\mu } =||P^{*}_{i}||- \frac{2}{\mu-1}\sum_{j=1}^{\mu-1}\calE_{P^*_{i,j},2} + \epsilon\calE_{Q_{i} ,\mu} 
%\leq B_i- \calE_{ Q^{*}_{i,j'}  ,2 } + \frac{1}{w-1}\calE_{Q_{i} ,w+1}.
\end{equation}

Now we do similar to what we did in Lemma \ref{lem:p2porienteering} to bound $\eps\calE_{Q_i,\mu}$.
We consider the $\mu$-jump of $Q_{i}$ based on $\langle v_{i,a_{1}},\cdots,v_{i,a_{\mu}} \rangle=\langle u^1_i,\ldots,u^\mu_i\rangle$, which are the same nodes in the $\mu$-jump of $P^*_i$ from which we obtained $Q_i$. By the definition of excess and how we obtained $Q_i$ from $P^*_i$:

\begin{equation}\label{eqn6}
\calE_{Q_{i} ,\mu}\leq || Q_{i} ||- || \langle v_{i,a_{1}},\cdots,v_{i,a_{\mu}} \rangle ||= \sum_{j=1}^{\mu-1}\calE_{ P^{*}_{i,j} ,2}  - \calE_{P^*_{i,j'},2} -\calE_{P^*_{i,j''},2} -\calE_{P^*_{i,\mu-1},2} \leq \sum_{j=1}^{\mu-1}\calE_{ P^{*}_{i,j} ,2}
\end{equation}

Thus, using (\ref{eqn5}) and (\ref{eqn6}) and noting that $\mu=\lfloor\frac{1}{\eps}\rfloor +1 $:

\[||Q^{'}_{i} || \leq ||P^*_i||-2\eps\sum_{j=1}^{\mu-1}\calE_{P^*_{i,j},2} + \eps\sum_{j=1}^{\mu-1}\calE_{P^{*}_{i,j} ,2} \leq ||P^*_i||- \eps\calE^*_{i,\mu}\leq
 ||J^*_{\mu}(P^*_i)||+(1-\eps)\calE^*_{i,\mu}.\]

To prove the upper bound required on $||Q'_i(v_i,v)||$ when considering any vertex $v$ left in $Q'_i$ (that was visited in $Q_i$ between $u^j_i,u^{j+1}_i$) we use the slightly more general version mentiond right after Theorem\ref{thm:muexcess}. Using that arguement and assuming that we take the $\mu$-jump of $P^*_i$ from which
we derived $Q_i$ instead, the extra length of $Q'_i$ compared to $Q_i$ at any dense cluster $C\in D$ is bounded by 
$\eps\sum_{j=1}^{\mu-1}\calE_{Q_{i,j},2}$. Noting that $\calE_{Q_{i,\mu-1},2}=0$
and definition of $D^2_i$ we can say: 

\begin{eqnarray*}
    ||Q'_i(v_i,v)|| &\leq& ||Q_i(v_i,v)||+\eps\sum_{j=1}^{\mu-2}\calE_{Q_{i,j},2}\\
    &=&  ||Q_i(v_i,v)||+\eps\sum_{j=1}^{\mu-2}\calE_{P^*_{i,j},2}\\
    &\leq& ||J^*_\mu(P^*_i(v_i,v'_i))||+\calE'_{i,\mu}-2\eps\sum_{j=1}^{\mu-2}\calE_{P^*_{i,j},2} +\eps\sum_{j=1}^{\mu-2}\calE_{P^*_{i,j},2}\quad\quad\quad\quad\mbox{using (\ref{eqn10})}\\
    &\leq&  ||J^*_\mu(P^*_i(v_i,v'_i))||+(1-\eps)\calE'_{i,\mu}.
\end{eqnarray*}
\end{proof}

{\bf Dynamic Programming:}
Now we describe how we can find a near optimum solution as guaranteed by Theorem \ref{the:mppstructure} using DP.
Recall the instance $\calI$ for deadline TSP as described before Theorem \ref{the:mppstructure} and let $P^*$ be an optimum solution for it.
Suppose we have guessed the vertices $\langle v_0,\ldots,v_m\rangle$ for $P^*$ where the $\mu$-excess
increases by at least $\alpha^i$ (recall $\alpha=1+\eps$). Also suppose for each $0\leq i<m$ we have guessed the vertices of the $\mu$-jump $J_i$ for $P^*_{v_i v_{i+1}}$ where 
$J_i:  v_i=u^1_i,u^2_i,u^3_i,\ldots,u^{\mu-1}_i,u^{\mu}_i=v_{i+1}$. We also we have guessed the visiting time of $v_i$'s ($0\leq i<m$); this means we have guessed $||J^*_\mu(P^*_{v_i v_{i+1}})||$
and also $\calE^*_{i,\mu}$. Observe that having all these values, we can find sets $N_{i,j}$ .
%Finally, we assume we have guessed $k_i=|P^*_{v_i,v_{i+1}}|$.
For each $i$, we can guess $k_i=|P^*_{v_iv_{i+1}}|$ and for those that $k_i<\mu^2$ we can guess $P^*_i$ exactly. 
Note that all these guesses can be done in time $(n\Delta)^{O(\mu^2h\delta)}$. We call this guessing {\em "Phase 1"}.
In the remaining (which is "Phase 2") we discuss how to find $Q'_i$'s for those $i$ for which $k_i\geq\mu^2$ that remain using DP. To simplify the notation we assume all $i$'s satisfy $k_i\geq\mu^2$ in what follows. Let $\calI_i$ be the instance defined before Theorem \ref{the:mppstructure} where the goal is to find a $v_i$-$v_{i+1}$-path that goes through $u^1_i,u^2_i,\ldots,u^\mu_i$ in this order; so it has $j$ leg where
the $j$'th leg is between $u^j_i,u^{j+1}_i$ and should use vertices in $N_{i,j}$. Our goal is to solve all $\calI_i$'s (where $k_i\geq\mu^2$) concurrently using a DP and hence find near optimum solutions $Q'_i$ as guaranteed by Theorem \ref{the:mppstructure}.

The DP is built on the $\gamma$-split tree $\Gamma$ we compute. Consider for any cluster node $C \in \Gamma$, any group $0\leq i<m$ and any leg $1\leq j<\mu$, the restriction of $Q'_{i,j}$ in the subgraph $G(C)$, which might be a collection of paths because $Q'_{i,j}$ can enter and exit $C$ multiple times. This causes $Q'_{i,j}$ be chopped up into multiple segments. For $C$, group $i$ and leg $j$ suppose that it is broken into $\sigma_{C,i,j}$ many pieces. 

\begin{definition}(multi-groups-legs multi-paths orienteering)
    Let $G=(V,E)$ be a graph, given groups $ 0\leq i<m$ and legs $1\leq j<\mu$, start and end node pairs $(s_{i,j,l},t_{i,j,l}), 1\leq l \leq \sigma_{i,j}$, budgets $B_{i}$ and subset $N_{i,j}$ as an instance of multi-group-legs multi-paths orienteering. The goal is to find a collection of paths $ Q_{i,j,l} $ such that $Q_{i,j,l}$ is a $s_{i,j,l}$-$t_{i,j,l} $ path, $ \sum_{j=1}^{\mu-1} \sum_{l=1}^{\sigma_{i,j}}||Q_{i,j,l}||\leq B_{i} $ and $ | \cup_{i=0}^{m-1} \cup_{j=1}^{\mu-1}( (\cup_{l=1}^{\sigma_{i,j}} Q_{i,j,l})\cap N_{i,j} )|$ is maximized.
\end{definition}

We define a subproblem in the DP as an instance of multi-groups-legs multi-paths orienteering  on $C$ with groups $ 0\leq i<m$ and legs $1\leq j<\mu$, start and end node pairs $(s_{i,j,l},t_{i,j,l}), 1\leq l \leq \sigma_{C,i,j}$, budgets $B_{C,i}$ and subset $N_{i,j}$. The goal is to find a collection of paths $ Q_{i,j,l} $ such that $Q_{i,j,l}$ is a $s_{i,j,l}$-$t_{i,j,l} $ path, $ \sum_{j=1}^{\mu-1} \sum_{l=1}^{\sigma_{C,i,j}}||Q_{i,j,l}||\leq B_{C,i} $ and $ | \cup_{i=0}^{m-1} \cup_{j=1}^{\mu-1}( (\cup_{l=1}^{\sigma_{i,j}} Q_{i,j,l})\cap N_{i,j} )|$ is maximized.\\We use $A[C,\{B_{C,i}\}_{0\leq i<m}, \{(s_{i,j,l},t_{i,j,l})_{l=1}^{\sigma_{C,i,j}} \}_{ 0\leq i<m; 1\leq j<\mu }]$ to denote the subproblem described above and the entry of the table to store the optimal value of the subproblem. We will show $\sigma_{C,i,j}$ is poly-logarithmic.

We compute the entries of this DP from bottom to up of $\Gamma$. The base cases are when $C$ has constant size $|C|=a$ thus such subproblem can be solved by exhaustive search.

In the recursion, consider any entry $A[C,\{B_{C,i}\}_{0\leq i<m}, \{(s_{i,j,l},t_{i,j,l})_{l=1}^{\sigma_{C,i,j}} \}_{ 0\leq i<m; 1\leq j<\mu}] $ and  any child split node $s$ of $C$ in $\Gamma$ and let $C_{1}, C_{2},\cdots, C_{g}$ be the children cluster nodes of $s$ in $\Gamma$. Recall $\pi_{s}$ is the corresponding partition of $C$ and $R_{\pi_{s}}$ is the set of bridge-edges in $C$ (crossing $\pi_{s}$). For each $C_{b}$ let $S_{b}$ be its portal set and recall $R'_{\pi_{s}}$ is the set of edges crossing $\pi_{s}$ only through $\{S_{b} \}$.

Note for each $i$ we consider two cases (the sparse case or the dense case). 
We will first guess two sets $I^{C}_{1}$ and $I^{C}_{2}$ that are disjoint and the union of them is $\{0,\cdots,m-1\}$. The intention is that
$I^{C}_{1}$ is the set containing those $i$'s that are (guessed to be) sparse w.r.t. $C$ and $I^{C}_{2}$ is the set of those $i$'s that are (guessed to be) dense w.r.t. $C$. For each $i\in I^{C}_{1}$ and for each $j$, we guess a subset of $R_{\pi_{s}} $, denoted as $E^{i,j}_{\pi_s}$ such that they are disjoint (for different $i,j$) and the size of $|\cup_{j=1}^{\mu-1} E^{i,j}_{\pi_s}|$ is at most $2\kappa'\frac{\log n}{\epsilon}$.
For $i\in I^{C}_{2} $ and for each $j$, we guess a subset of $R'_{\pi_{s}} $, also denoted as $E^{i,j}_{\pi_s}$ that size of $|\cup_{j=1}^{\mu-1} E^{i,j}_{\pi_s}|$ is at most $(\frac{16\kappa' \delta}{\epsilon})^{2\kappa} $ (since these are portal edges they need not be from $N_{i,j}$). 
For each $i$ we guess $B_{C_{1},i},\cdots,B_{C_{g},i}, $ such that $\sum_{j=1}^{g} B_{C_{j},i} + \sum_{j=1}^{\mu-1}\sum_{(u,v)\in E^{i,j}_{\pi_{s}}} d(u,v) =B_{C,i}$. 

We show how to guess start-end pairs $\{(s_{i,j,l},t_{i,j,l})_{l=1}^{\sigma_{C,i,j}} \}_{ 0\leq i<m; 1\leq j<\mu}$  for each $C_{b}$ and  check the consistency of them: for each $E^{i,j}_{\pi_s} $ and for each edge in $E^{i,j}_{\pi_{s}} $, we guess it is in which one of the $\sigma_{C,i,j}$ paths with start-end pair $(s_{i,j,l},t_{i,j,l})_{l=1}^{\sigma_{C,i,j}} $ and for each path with start-end pair $(s_{i,j,l},t_{i,j,l})$ we guess the order of the guessed edges appearing on the path. Specifically speaking, let $e_{1},e_{2}\cdots,e_{w}$  be the edges guessed in the path with start-end pair $(s_{i,j,l},t_{i,j,l})$ appearing in this order. Let $ C_{a_{1}},C_{a_{2}},\cdots,C_{a_{w+1}} $ be the children cluster nodes of $s$ that the path encounters following $ e_{1},e_{2}\cdots,e_{w} $, i.e. $e_{1}$ crosses $C_{a_{1}}$ and $C_{a_{2}}$,  $e_{2}$ crosses $C_{a_{2}}$ and $C_{a_{3}}$, $\cdots$,  $e_{w}$ crosses $C_{a_{w}}$ and $C_{a_{w+1}}$. Then we set $s_{i,j,l}$ and the endpoint of $e_{1}$ in $C_{a_{1}}$ to be a start-end pair in group $i$ and leg $j$ in $C_{a_{1}}$, the endpoint of $e_{1}$ in $C_{a_{2}}$  and the endpoint of $e_{2}$ in $C_{a_{2}}$ to be a start-end pair in group $i$ and leg $j$ in $C_{a_{2}}$, $\cdots$, the endpoint of $e_{w}$ in $C_{a_{w+1}}$ and $t_{i,j,l}$ to be a start-end pair in group $i$ and leg $j$ in $C_{a_{w+1}}$. By doing so we generate start-end pairs in group $i$ and leg $j$ for each $C_{b} $ and we sort them based on their ordering in $s_{i,j,l}$-$t_{i,j,l}$ path. This defines $\sigma_{C_{b},i,j}$ start-end pairs for each $C_{b}$.
Formally, to compute $A[C,\{B_{C,i}\}_{0\leq i<m}, \{(s_{i,j,l},t_{i,j,l})_{l=1}^{\sigma_{C,i,j}} \}_{ 0\leq i<m; 1\leq j<\mu}] $:
\begin{itemize}
    \item Consider any child split node $s$ of $C$ and let $C_{1}, C_{2},\cdots, C_{g}$ be the children cluster nodes of $s$.
    \item  Let $\pi_{s}$ be the corresponding partition of $C$ and $R_{\pi_{s}}$ be the set of bridge-edges of $\pi_{s}$. For each $C_{b}$ let $S_{b}$ be its portal set and let $R'_{\pi_{s}}$ be the set of portal-edges of $\pi_{s}$.
    \item Guess sets $I^{C}_{1}$ and $I^{C}_{2}$ such that they are disjoint and the union of them is $\{0,\cdots,m-1\}$.

    \item For $i\in I^{C}_{1} $, we guess a subset of $R_{\pi_{s}} $, denoted as $E^{i,j}_{\pi_s}$ such that they are disjoint (for different $i,j$) and the size of $|\cup_{j=1}^{\mu-1} E^{i,j}_{\pi_s}|$ is at most $2\kappa'\frac{\log n}{\epsilon}$. For $i\in I^{C}_{2} $ and for each $j$, we guess a subset of $R'_{\pi_{s}} $, also denoted as $E^{i,j}_{\pi_s}$ that size of $|\cup_{j=1}^{\mu-1} E^{i,j}_{\pi_s}|$ is at most $(\frac{16\kappa' \delta}{\epsilon})^{2\kappa} $  

    \item For each $i$ we guess $B_{C_{1},i},\cdots,B_{C_{g},i}, $ such that $\sum_{j=1}^{g} B_{C_{j},i} + \sum_{j=1}^{\mu-1}\sum_{(u,v)\in E^{i,j}_{\pi_{s}}} d(u,v) =B_{C,i,\ell}$.
    \item For each $E^{i,j}_{\pi_s} $ and for each edge in $E^{i,j}_{\pi_{s}} $, we guess it is in which one of the $\sigma_{C,i,j}$ paths with start-end pair $(s_{i,j,l},t_{i,j,l})_{l=1}^{\sigma_{C,i,j}} $ and for each path with start-end pair $(s_{i,j,l},t_{i,j,l})$ we guess the order of the guessed edges appearing on the path. as described above. We generate start-end pairs in group $i$ and leg $j$ for each $C_{b} $ accordingly. Then:
    \item We set $A[C,\{B_{C,i}\}_{0\leq i<m}, \{(s_{i,j,l},t_{i,j,l})_{l=1}^{\sigma_{C,i,j}} \}_{ 0\leq i<m; 1\leq j<\mu}] =\\\max\sum_{b=1}^{g} A[C_{b}, \{B_{C_{b},i}\}_{0\leq i<m},  \{(s_{i,j,l},t_{i,j,l})_{l=1}^{ \sigma_{C_{b},i,j}} \}_{ 0\leq i<m; 1\leq j<\mu}  ] $,  where the maximum is taken over all tuples $(s,\{ B_{C_{b},i}\}_{1\leq b\leq g;0\leq i<m}, \{(s_{i,j,l},t_{i,j,l})_{l=1}^{\sigma_{C_{b},i,j}}\}_{1\leq b\leq g;0\leq i<m;1\leq j<\mu} )$ as described above.
\end{itemize}

%As said earlier, we are going to guess all $v_0,v_1,\ldots,v_m$
%and also $u^1_i,\ldots,u^{\mu-2}_i$ vertices (for $0\leq i< m$)
%and also guess $\calE^*_{i,\mu}$.
%These guesses will give us all budgets $B_i=||J^*_\mu(P^*_{v_iv_{i+1}})||+\calE^*_{i,\mu}$.
The goal is to compute $A[c_{0},\{ B_{i}-\eps \calE^{*}_{i,\mu} \}_{0\leq i<m},\{(u_i^{j},u_i^{j+1})\}_{0\leq i<m;1\leq j<\mu} ]$.

Now we analyze the running time of DP. First we show the running time of computing one entry of the DP table is $n^{O((\frac{\delta}{\epsilon})^{2\kappa+2} )} $. In the recursion, for a cluster node $C$, there are  $\gamma=n^{3h\delta}$ children split nodes of $C$ to consider. For a certain split node $s$, let $C_{1},\cdots,C_{g}$ be the children cluster nodes of $s$.
There are $2^{m+1}$ guessing for $I^{C}_{1},I^{C}_{2}$ because they are disjoint and the union of them is $\{0,\cdots,m-1\}$. For $\{E^{i,j}_{\pi_{s}}\}_{0\leq i<m;1\leq j\leq \mu=1}$ : for each group $i$ and leg $j$, if $i\in I^{C}_{1}$, then because $|\cup_{j=1}^{\mu-1} E^{i,j}_{\pi_{s}}|\leq 2\kappa'\frac{\log n}{\epsilon}  $, there are at most $n^{4\kappa'\frac{\log n}{\epsilon} } $ many possible $E^{i,j}_{\pi_{s}}$ to consider; if $i \in I^{C}_{2}$ then because $E^{i,j}_{\pi_{s}}\subset R'_{\pi_{s}}$ and $ | R'_{\pi_{s}} |\leq (\frac{16\kappa' \delta}{\epsilon})^{2\kappa}$ in this case, there are at most $2^{ (\frac{16 \kappa' \delta}{\epsilon})^{2\kappa}  } \leq n^{ ( \frac{16\kappa' \delta}{\epsilon})^{2\kappa} } $
many possible $E^{i,j}_{\pi_{s}}$ to consider. Hence there are at most $[ n^{ ( \frac{16\kappa' \delta }{\epsilon})^{2\kappa}}]^{(m+1)\mu}$ many possible $E^{i,j}_{\pi_{s}}$ to consider for all $i$ and all $j$. There are at most $ [( n\Delta_{G} )^{g}]^{m+1} $ guessings for $\{B_{C_{1},i}, \cdots,B_{C_{g},i}\}$ such that  $ \sum_{b=1}^{g} B_{C_{b},i} + \sum_{j=1}^{\mu-1} \sum_{(u,v)\in E^{i,j}_{s}} d(u,v) =B_{C,i} $ for all $i$. 
To generate $\{(s_{i,j,l},t_{i,j,l})_{l=1}^{\sigma_{C,i,j}} \}_{ 0\leq i<m; 1\leq j<\mu}$  for each $C_{b}$: for each group $i$ and leg $j$ and for each edge in $E^{i,j}_{s}  $, we guess it is in which one of $\sigma_{C,i,j}$ path with start-end pair $(s_{i,j,l},t_{i,j,l})_{l=1}^{\sigma_{c,i,j}}$ and for each path with start-end pair $(s_{i,j,l},t_{i,j,l})$ we guess the order of the edges appearing, which is at most $ (|E^{i,j}_{\pi_{s}}|!|E^{i,j}_{\pi_{s}}|^{\sigma_{C,i,j}}) $ guessings. Note at each recursion for group $i$ and leg $j$ it may increase at most $|E^{i,j}_{\pi_{s}} |$ the number of start-end pairs and the depth of the recursion is $\delta$. Thus $ \sigma_{C,i,j}\leq \delta|E^{i,j}_{\pi_{s}} |$  and the total guessings for all $i$ and all $j$ is at most $[|E^{i}_{\pi_{s}}|! |E^{i}_{\pi_{s}}|^{\sigma_{c,i}}]^{(m+1)\mu}\leq ( (\frac{16\kappa' \delta}{\epsilon})^{2\kappa} ! (\frac{16\kappa' \delta}{\epsilon})^{2\kappa \delta (\frac{16\kappa' \delta}{\epsilon})^{2\kappa}} )^{(m+1)\mu}\leq n^{O((\frac{\delta}{\epsilon})^{2\kappa+2})} $.

We show the size of the dynamic programming table is at most $n^{O((\frac{\delta}{\epsilon})^{2\kappa+2})}$. Recall the entry of the table is  $A[C,\{B_{C,i}\}_{0\leq i<m}, \{(s_{i,j,l},t_{i,j,l})_{l=1}^{\sigma_{C,i,j}} \}_{ 0\leq i<m; 1\leq j<\mu}]$. For $C$, there are at most $( n^{3h\delta} 2^{\kappa})^{\delta} $ cluster nodes in $\Gamma$ . For $B_{C,0},\cdots,B_{C,m}$, there are at most $(n\Delta_{G})^{m+1}$ many possible value to consider. For $\{(s_{i,j,l},t_{i,j,l})_{l=1}^{\sigma_{C,i,j}} \}_{ 0\leq i<m; 1\leq j<\mu} $, there are at most $n^{O((\frac{\delta}{\epsilon})^{2\kappa+2})} $ possible start-end pairs to consider.

Therefore, computing the whole DP table and finding $\{Q^{'}_{i,j}\}_{0\leq i<m; 1\leq j<\mu}$ as in Theorem \ref{the:mppstructure} takes at most $n^{O((\frac{\delta}{\epsilon})^{2\kappa+2})} $ time. 
Considering the time spent guessing in Phase 1, the total time 
is $n^{O((\frac{\delta}{\eps})^{2\kappa+2})} $.

As mentioned, we compute $A[c_{0},\{ B_{i}-\eps \calE^{*}_{i,\mu} \}_{0\leq i<m},\{(u_i^{j},u_i^{j+1})\}_{0\leq i<m;1\leq j<\mu} ]$ for all guesses of $v_0,v_1,\ldots,v_m$, $u^1_i,\ldots,u^{\mu-2}_i$ (for $0\leq i< m$), 
and also $\calE^*_{i,\mu}$ (which yields $B_i$'s as well as $N_{i,j}$'s).
For each solution we consider $Q'$ that is the path obtained by concatenating $Q'_{0},Q'_{1},\cdots,Q'_{m}$
and check if all deadlines are respected. We pick the feasible solution with maximum $|Q'|$. By Theorem \ref{the:mppstructure} for at least one of them we have a $(1-4\eps)$-approximation for $P^*$. Replacing $\eps=\eps/4$ in all the calculations this implies the the proof
of Theorem \ref{thm:deadlineDBL}.

{\bf Bicriteria approximation for distances in $\QQ^+$:} We can adapt the analysis of proof of Theorem \ref{thm:deadlineDBL} to work when
distances have fractional values, i.e. are from $\QQ^+$ instead of integers. In this setting 
we bucket the distances (and budgets) based on powers of $\lambda=(1+\eps^2/\delta)$. More precisely, for any value $x$ let $L(x)$ (which stands for rounded down value) be the nearest power of $\lambda$ that is at most $x$ and $R(x)$ (which stands for rounded up value) be the nearest power of $\lambda$ that is at least $x$.
Whenever we consider distance values in our calculations we consider taking $L(.)$ and whenever we consider budget value we consider taking $R(.)$. A feasible solution after rounding down distance and rounding up budgets  remains feasibility. Since the height of $\Gamma$, and hence the depth of our DP recursion is $O(\delta)$, the number of additions of distances and budgets (hence rounding error accumulation) is at most $O(\delta)$. Thus, if we find a solution using DP, the actual distances might be off by $\lambda^{O(\delta)}\leq O(\eps)$ factor. So our solution at the end might be violating the deadlines by at most $O(\eps)$-factor.
The size of the table decreases slightly as for each $B_{c,i}$
the number of possible values will be $\log_\lambda B_{C,i}\leq\log(n\Delta)\delta/\eps^2\leq\delta^2\log n/\eps^2$.
So the running time for DP is still upper bounded by $n^{O((\delta/\eps)^{2\kappa+2})}$.

\bibliographystyle{plain}
\bibliography{references}

\end{document}